\documentclass[fleqn, 12 pt]{article}
\usepackage{hhline}
\usepackage[hidelinks]{hyperref}
\usepackage{pdflscape}
\usepackage{threeparttable}
\usepackage{comment}
\usepackage{tabularray}
\usepackage{color}
\usepackage{threeparttable}
\usepackage[bottom]{footmisc}
\usepackage[papersize={8.3in,11.7in},vmargin={1.0in,1.0in}, hmargin={1.0in,1.0in}]{geometry}
\usepackage{times}

\usepackage[round]{natbib}
\usepackage{authblk}
\usepackage{soul}

\usepackage{graphicx, rotating}
\usepackage[utf8]{inputenc}
\usepackage{enumerate}
\usepackage{euscript}
\usepackage{bbm}
\usepackage{amsmath, amsthm,amsfonts, amssymb}
\usepackage{extarrows}
\usepackage{tikz}
\usetikzlibrary{calc}
\hypersetup{colorlinks,linkcolor={blue},citecolor={blue},urlcolor={blue}}
\usetikzlibrary{positioning}
\usetikzlibrary{arrows.meta}
\usepackage{caption}
\usepackage{subcaption}

\usepackage{pgfplots}
\usepgfplotslibrary{dateplot}
\usepackage{pdfpages}

\usepackage{array}
\usepackage{multirow}
\usepackage{colortbl}
\usepackage[font=small]{caption}

\usepackage{float}
\newtheorem{theorem}{Theorem}

\newtheorem{proposition}{Proposition}
\newtheorem{corollary}{Corollary}

\newtheorem{observation}{Observation}

\newtheorem*{example*}{Example}

\usepackage{setspace}

\usepackage[labelsep=period, skip=5pt]{caption}

\usepackage{verbatim, booktabs, multirow, hyperref}
\usepackage{bbm}
\usepackage{lmodern}
\usetikzlibrary{intersections}

\usepackage{newtxtext,newtxmath}
\usepackage{url}
\usepackage{ulem}
\usepackage{dcolumn}
\newcolumntype{d}[1]{D{.}{.}{#1}}

\usepackage{tikz}
\usetikzlibrary{positioning}
\usepackage{tikz-3dplot}
\usetikzlibrary{arrows, patterns, decorations.pathmorphing,backgrounds,positioning,fit,petri,trees}
\usepackage{pgfplots}
\usepackage{pgfplotstable}

\usetikzlibrary{arrows}
\usetikzlibrary{calc}
\usetikzlibrary{positioning}
\usetikzlibrary{arrows.meta}
  
\usetikzlibrary{decorations.pathreplacing,calligraphy}
\usetikzlibrary{shapes.geometric}

\pgfplotsset{compat=1.18}
\usepackage{enumitem}
\setlist[itemize]{leftmargin=*}

\title{Network games with heterogeneous players}

\author{Wenjie Cao\footnote{Laboratory of Mathematics and Complex Systems, Ministry of Education, School of Mathematical Sciences, Beijing Normal University, 100875 Beijing, P. R. China. Email: wenjiecao@mail.bnu.edu.cn.}
\quad\quad Angel Sánchez\footnote{Grupo Interdisciplinar de Sistemas Complejos (GISC), Departamento de Matemáticas, Universidad Carlos III de Madrid, Leganés 28911, Spain; Instituto de Biocomputación y Física de Sistemas Complejos (BIFI), Universidad de Zaragoza, Zaragoza 50018, Spain. Email: anxo.sanchez@gmail.com}
\quad\quad Boyu Zhang\footnote{Laboratory of Mathematics and Complex Systems, Ministry of Education, School of Mathematical Sciences, Beijing Normal University, 100875 Beijing, P. R. China. Corresponding author. Email: zhangby@bnu.edu.cn.}}

\date{\today }

\begin{document}

\pagenumbering{gobble}

\maketitle

\begin{abstract}
Real social and economic networks involve individuals with diverse incentives, yet most studies of network games assume homogeneous preferences or few player types. We introduce a general framework for binary-choice network games with fully heterogeneous payoff structures. We first show that any such game can be transformed into an equivalent one with conformist, rebel, and stubborn archetypes, preserving equilibria and best-response trajectories. We then establish sufficient conditions for pure-strategy Nash equilibrium existence and convergence of best-response dynamics on arbitrary networks, while proving that equilibria almost surely vanish in large sparse random networks. We further develop a deterministic approximation approach that predicts evolutionary trends and equilibrium strategy frequencies from network homophily and heterophily patterns, without computing equilibria explicitly. Extending the framework to limited information, we prove that dynamics converge either to a unique limited-information equilibrium or to a unique stationary distribution, and we derive necessary and sufficient conditions for the existence of the limited-information equilibrium. We validate our predictions using Prisoner’s Dilemma games on real social networks that incorporate heterogeneous altruism and peer influence. These findings together provide a unified framework for equilibrium existence, evolutionary dynamics, and equilibrium outcome prediction in heterogeneous network games.

\vspace{.25cm}
\noindent \textbf{Keywords:}  Network games, Heterogeneous payoff, Best-response dynamics, Limited information
\\
\textbf{JEL: C72, C73, D83, D85} 
\end{abstract}

\setstretch{1.5}
\pagenumbering{arabic}

\thispagestyle{empty}
\newpage

\clearpage

\section{Introduction}
\setcounter{page}{1}
\subsection{Research background}
Network games, a central domain within game theory, offer a principled framework for modeling complex systems of interconnected agents \citep[see e.g.][]{jackson2008social}. Their analytical power has made them indispensable for studying a wide range of socio-economic and biological phenomena, including the evolution of cooperation \citep[see e.g.][]{nowak1992evolutionary,ohtsuki2006simple,santos2008social,allen2017evolutionary}, 
competition in networked and two-sided markets \citep[see e.g.][]{bim2019,han2025coordination}, the diffusion of innovations \citep[see e.g.][]{jackson2007diffusion,young2011dynamics}, and the propagation of financial risks \citep[see e.g.][]{banerjee2013diffusion,elliott2014financial}.

In network games, agents are represented as nodes and their interactions as edges. Each agent’s payoff depends on its own strategy and those of its neighbors \citep[see e.g.][]{jackson2008social}. For evolutionary games on static networks, central analytical challenges include the issues of (i) the existence and uniqueness of pure strategy Nash equilibria (PNE), (ii) the short-run behavior of the evolutionary process, and (iii) its long-run outcomes. These questions have been systematically examined when players are homogeneous in payoffs. Potential game methods are widely used to establish the existence of PNE and to characterize the convergence of best response dynamics \citep[see e.g.][]{monderer1996potential,bramoulle2014strategic,wu2019potential}. The short-run dynamics of stochastic update rules, such as best response, pairwise imitation, and death-birth updating, can often be approximated via systems of differential equations \citep[see e.g.][]{ohtsuki2006replicator,tarnita2011multiple,wang2024evolutionary,PEI2024dynamic}. Furthermore, extensive work has analyzed long-run outcomes by computing fixation probabilities of mutant strategies \citep[see e.g.][]{ohtsuki2006simple,allen2017evolutionary,mcavoy2020social} 
and stationary distributions of stochastic evolutionary processes \citep[see e.g.][]{kandori1993learning,young1993evolution}.

In real-world applications, agents are often heterogeneous in multiple dimensions, including resources, productivities, social preferences, and functional roles, which shape their strategic incentives. The existing literature on network games with heterogeneous players can be grouped into two main strands. The first strand allows agents to have different payoff functions but assumes they share similar social interaction preferences.  Typical examples are games of strategic complements and strategic substitutes \citep[see e.g.][]{galeotti2010network,HERNANDEZ201356,bramoulle2014strategic,jackson2015games_networks,bramoulle2016games_networks,ramazi2016networks}. These models accommodate payoff heterogeneity across individuals, yet they restrict every player’s relative payoff from increasing an action to be either uniformly increasing (strategic complements) or uniformly decreasing (strategic substitutes) in the set of neighbors taking the action. The second strand allows heterogeneous preferences by categorizing agents into distinct behavioral types, such as conformists, rebels, coordinators, anti-coordinators, and stubborn agents, and assigns the same payoff function to all agents within a type \citep[see e.g.][]{acemouglu2013opinion, yildiz2013binary, zhang2018fashion, Broere2019, cao2019dynamic, stewart2019information,cao2024discrete}. Besides, a limited number of studies permit simultaneous heterogeneity in both payoffs and preferences, but they are typically confined to specific game classes \textcolor{blue}{\citep{KARP2007global,bramoulle2014strategic,HOFFMANN2019Global}}. In summary, network games with fully heterogeneous payoffs, where each player may have different preference parameters, remain largely unexplored. In particular, systematic approaches to the three fundamental questions mentioned above in this general setting are still lacking, posing a significant methodological gap.

\subsection{Our contribution}
In this paper, we consider binary choice network games in which players can have different preferences or payoff functions (hereafter referred to as HP games), where each player’s payoff depends on their own choice and the strategies of neighbors. We provide a unified framework for investigating the existence of pure strategy Nash equilibria (PNE for short), the short-run evolutionary behavior of the best response dynamics, and the strategy frequencies at stable equilibria on arbitrary networks. 

As a first step, we show that any HP game can be reformulated as an equivalent network game with three types of players, conformists, rebels, and stubborn agents (CRS game from now on), preserving both PNE and best response trajectories of the original one (Theorem \ref{th1}). In the reformulation, we introduce stubborn neighbors to original coordinators and anti-coordinators, and maintain a one-to-one correspondence between conformists and coordinators, and between rebels and anti-coordinators. 

Then, using the potential and partial potential game approaches \citep[]{monderer1996potential,zhang2018fashion}, we establish sufficient conditions for the existence of PNE and convergence of asynchronous best response dynamics in CRS games, and extend these results to the original HP games (Theorem \ref{th2} and Corollary \ref{co1}). Specifically, a CRS game admits a PNE if one of the following conditions holds: (i) the game does not have conformist-rebel (CR) edges; (ii) each conformist has enough conformist neighbors; (iii) each rebel has enough rebel neighbors. This result is a direct extension of \citet[]{zhang2018fashion}, which only considers games with conformists and rebels. Regarding non-existence of PNE, We prove that equilibria almost surely disappear in large random sparse networks (Theorem \ref{th3}). 

Testing the existence of a PNE for a CRS or HP game has been shown to be NP-hard \citep[]{cao2014fashion}. Rather than conducting a direct PNE analysis, we study the asymptotic behavior of the asynchronous best response dynamics. Since each PNE corresponds to an absorbing state of these dynamics, their asymptotic behavior can characterize the properties of PNE. In Section \ref{sec4}, we develop an approximation method that transforms the asynchronous best response dynamics into a deterministic system of ordinary differential equations (ODE), applicable to any network structure. This method enables the estimation of equilibrium strategy frequencies based on network homophily and heterophily patterns without explicitly computing PNE (Theorem \ref{th4}). More importantly, our method can predict the evolutionary trends of the asynchronous best response dynamics even when the game does not have a PNE (which implies that the dynamics do not converge). We then fully characterize the dynamic behaviors of the ODE system. We find that few conformist-rebel edges in CRS games, or coordinator-anti-coordinator edges in HP games, drive the ODE trajectories toward boundary fixed points (which typically correspond to PNE), whereas numerous such edges can lead to non-convergent trajectories and the emergence of periodic oscillations (Corollaries \ref{co2} and \ref{co3}).

We further extend the analysis to a limited information setting in which each player observes each neighbor’s strategy with a positive probability less than one \citep{oyama2015sampling}. In this case, the HP game has at most one PNE (referred to as an L‑PNE) from which we derive the necessary and sufficient existence conditions (Theorem \ref{th4}). At an L‑PNE, all coordinators choose the same action, whereas anti‑coordinators are divided between the two actions. In addition, both asynchronous and synchronous best response dynamics converge to the L‑PNE when it exists, and otherwise to a unique stationary distribution whose strategy frequencies approximate a stable fixed point of the associated ODE system (Corollary \ref{co4}).

To verify the validity of our approach, we apply our framework to empirical network data, examining Prisoner’s Dilemma games involving players with heterogeneous altruistic preferences and peer influence, and find that our approximation accurately predicts cooperation frequencies in such settings (see Section \ref{sec6}). We further discuss possible mechanisms for promoting cooperation in this game. First, the upper bound on the proportion of cooperation can be raised by reducing the prevalence of stubborn defectors. Second, this upper bound becomes stable when conformists connect more to stubborn cooperators than to stubborn defectors, whereas rebels display the opposite connectivity pattern.

Finally, we note that our analytical framework is general and several key findings, such as existence conditions for PNE and L‑PNE, can be naturally extended to multi‑strategy games (see Section \ref{sec7.2}). In summary, our results demonstrate how network topology and heterogeneity in players’ preferences jointly shape equilibrium outcomes, offering a systematic framework for analyzing cooperation, coordination, and conflict in heterogeneous populations.

\subsection{Related literature}
 
Our work is closely related to network games with strategic complements and strategic substitutes. Most existing studies assume that players have similar preferences, either strategic complements or strategic substitutes, and focus on the existence and uniqueness of PNE, and long-run outcomes of evolutionary dynamics \citep[see e.g.][]{galeotti2010network,HERNANDEZ201356,bramoulle2014strategic,ramazi2016networks}. Our HP game is most closely related to \cite{ramazi2016networks}, who also consider binary choice network games and allow payoffs to differ across individuals. However, their analytical framework retains an underlying symmetry in preferences for strategic complements or substitutes. Relatively few studies examine the coexistence of strategic complements and substitutes, and these are often confined to specific game classes such as public goods games and global games\citep{KARP2007global, bramoulle2014strategic, HOFFMANN2019Global}. Systematic approaches for analyzing properties of PNE and evolutionary dynamics in general setting are still lacking.


Our CRS game is a network game with three player types. Most prior work focuses on games with only two types. For example, \cite{zhang2018fashion} and \cite{cao2019dynamic} analyze the fashion game of \cite{jackson2008social}, which features conformists and rebels. In contrast, \cite{cao2024discrete} study network games with conformists and stubborn agents. The impact of stubborn agents on strategy evolution in binary choice games has also been examined by \cite{acemouglu2013opinion}, \cite{yildiz2013binary}, and \cite{stewart2019information}. A recent study by \cite{PEI2024dynamic} further generalizes the analysis to networks with two arbitrary agent types. However, network games with three types of players remain largely unexplored.

\subsection{Paper structure}
The remainder of the paper is organized as follows. Section \ref{sec2} introduces the HP and CRS game models. Section \ref{sec3} analyzes their pure-strategy Nash equilibria. Section \ref{sec4} explores the evolutionary trends and equilibrium strategy frequencies of the best response dynamics. Section \ref{sec5} investigates the case of limited information. Section \ref{sec6} applies our method to Prisoner’s Dilemma games on real social networks. Section \ref{sec7} concludes, discusses extensions, and examines the limitations of this study. All proofs and technical details are provided in the \textit{Appendix}.

\section{MODEL}\label{sec2}
\subsection{Network games with heterogeneous payoffs}

We consider an undirected network $G=(N,E)$, where the nodes $N=\{1,2, \ldots, n\}$ correspond to players and each edge in the set $E \subseteq N \times N$ represents a $2$-player game between neighboring players. Each player $i\in N$ chooses pure strategies from a binary set $A=\left\{0, 1\right\}$ and receives a payoff upon completion of the game according to the matrix $U_i$:
$$
\bordermatrix{
   & 0 & 1 \cr
0 & a_i  & b_i \cr
1 & c_i & d_i \cr
}, \quad a_i, b_i, c_i, d_i \in \mathbb{R}. 
$$
The key feature of the model is that players have different payoff matrices \citep[see e.g.][]{ramazi2016networks}. Therefore, we call the system $G=(N,E,A,U)$ a network game with heterogeneous payoffs (HP game for short).

We now introduce the evolutionary process. The evolutionary dynamics take place over a sequence of discrete time
$t=0,1,2, ...$. Let $x_i(t) \in A$ denote the strategy of player $i$ at time $t$, and denote the number of neighbors of player $i$ choosing strategy $0$ and $1$ at time $t$ by $n_i^0(t)$ and $n_i^1(t)$, respectively. The total payoffs to player $i$ at time $t$ are accumulated over all neighbors, and are therefore equal to $a_i n_i^0(t) + b_i n_i^1(t)$ when $x_i=0$, or $c_i n_i^0(t) + d_i n_i^1(t)$ when $x_i=1$. 

In asynchronous (myopic) best response dynamics, one player at each time becomes active and chooses a single action to play against all neighbors. At time $t+1$, the active player chooses the action that achieves the highest total payoff based on the strategy profile at time $t$:
$$
x_i(t+1)=\left\{\begin{array}{c}
0, \text { if } a_i n_i^0(t) + b_i n_i^1(t)>c_i n_i^0(t) + d_i n_i^1(t) \\
1, \text { if } a_i n_i^0(t) + b_i n_i^1(t)<c_i n_i^0(t) + d_i n_i^1(t) \\
x_i(t), \text { if } a_i n_i^0(t) + b_i n_i^1(t)=c_i n_i^0(t) + d_i n_i^1(t)
\end{array}\right..
$$
In the case that the two strategies result in equal payoffs, both strategies are best responses and we assume players keep their current strategy, since they have no incentive to deviate. We note that each pure strategy Nash equilibrium (PNE) corresponds to a fixed point of the above best response dynamics. Thus, the asymptotic behavior of the best response dynamics can be used to characterize the properties of PNE.

It is convenient to rewrite the best response dynamics above in terms of the number of neighbors playing each strategy, making them equivalent to a linear-threshold model \citep[see e.g.][]{granovetter1978threshold,kleinberg2007cascading,ramazi2016networks}: 

$$
x_i(t+1)=\left\{\begin{array}{c}
0, \text { if } \operatorname{sgn}(\delta_i) n_i^0(t) > \operatorname{sgn}(\delta_i)\tau_i n_i \\
1, \text { if } \operatorname{sgn}(\delta_i) n_i^0(t) <\operatorname{sgn}(\delta_i) \tau_i n_i \\
x_i(t), \text { if } \operatorname{sgn}(\delta_i) n_i^0(t) =\operatorname{sgn}(\delta_i) \tau_i n_i 
\end{array}\right.,
$$
where $\operatorname{sgn}(\cdot)$ is the sign function, $\delta_i=a_i-c_i+d_i-b_i$ determines the type of player $i$, $\tau_i= \frac{d_i-b_i}{\delta_i}$ is the threshold of player $i$ when $\delta_i \neq 0$, and $n_i$ is the degree of player $i$.
We thus see that if $\delta_i > 0$ and $\tau_i \in (0,1)$, player $i$ adopts strategy $0$ when a sufficient fraction of her neighbors play  strategy $0$; we refer to such a player as a coordinator. If $\delta_i < 0$ and $\tau_i \in (0,1)$, player $i$ adopts strategy~$1$ when a sufficient fraction of her neighbors play strategy $0$; we refer to her as an anti-coordinator. If $\delta_i = 0$ or $\tau_i \notin (0,1)$, player $i$ has a dominant strategy, and we refer to her as a stubborn agent (see Figure \ref{fig1}A and B for an example). We denote the sets of coordinators, anti-coordinators, stubborn agents with dominant strategy 0, and stubborn agents with dominant strategy 1 by $Co$, $ACo$, $S_0$, and $S_1$, respectively.

\subsection{Network games with three types of players} Since every player in the HP game has different payoff functions, it is difficult to directly analyze its PNE. To address this issue, we introduce a relatively simpler game, network games with three types of players (CRS game for short): conformists who prefer to coordinate with the majority, rebels who prefer to deviate from the majority, and stubborn agents who never change their initial strategy. With the above notation, it can be seen that a CRS game is a special HP game in which all coordinators and anti-coordinators have the same threshold $\tau_i = \frac{1}{2}$. We next show that for each HP game, we can construct an equivalent CRS game in the sense that the two games have identical PNE and strategy evolution trajectories. 

Formally, a CRS game can be represented by a system $G=(C,R,S,E,A,V)$, where $C$, $R$, and $S$ are the sets of conformists, rebels, and stubborn agents, respectively, $E$ is the set of edges, $A=\{0,1\}$ is the set of strategies, and $V=\left(u_i(\cdot)\right)_{i \in N}$ is the set of payoff functions, with
$$
u_i(\boldsymbol{x}) =
\begin{cases}
n_i^{x_i}, & \text{if } i \in C \\
n_i - n_i^{x_i}, & \text{if } i \in R \\
0, & \text{if } i \in S
\end{cases},
$$
where $n_i^{x_i}$ is the number of neighbors of player $i$ with strategy $x_i$ (i.e., neighbors using the same strategy as player $i$). 

We begin by defining equivalence between an HP game and a CRS game. The two games are equivalent if the following two conditions hold: (i) There exists a one-to-one correspondence between $Co$ and $C$, and between $ACo$ and $R$. (ii) Each coordinator (or anti-coordinator) $i$ and her corresponding conformist (or rebel) $i'$ have the same best response strategy when $x_j=x_{j'}$ for all other player $j \in Co \cup ACo$ and her corresponding player $j'\in C \cup R$. Thus, an HP game and its equivalent CRS game have identical PNE in the sense that, at corresponding equilibria, each coordinator and anti-coordinator in the HP game adopts the same strategy as the corresponding conformist and rebel in the CRS game. In addition, if coordinators and anti-coordinators in the HP game have the same initial strategy and active order with their corresponding conformists and rebels, then the two games yield identical strategy evolution trajectories under the asynchronous best response dynamics. We can now state the key result in order to use this equivalence to characterize PNE of HP games:

\begin{theorem}
For any HP game, we can construct an equivalent CRS game by adding stubborn neighbors to coordinators and anti-coordinators.
\label{th1}
\end{theorem}

The above theorem indicates that, for each HP game, there is an equivalent CRS game (see Figure \ref{fig1} for an example). Thus, we can characterize the PNE and the evolutionary dynamics of an HP game by analyzing its equivalent CRS game. 
The proof of the theorem can be found in \textit{Appendix}, section S1A. In the next section, we use this result to analyze HP games in general.

\begin{figure}[H]
\centering
\includegraphics[width=0.8\linewidth]{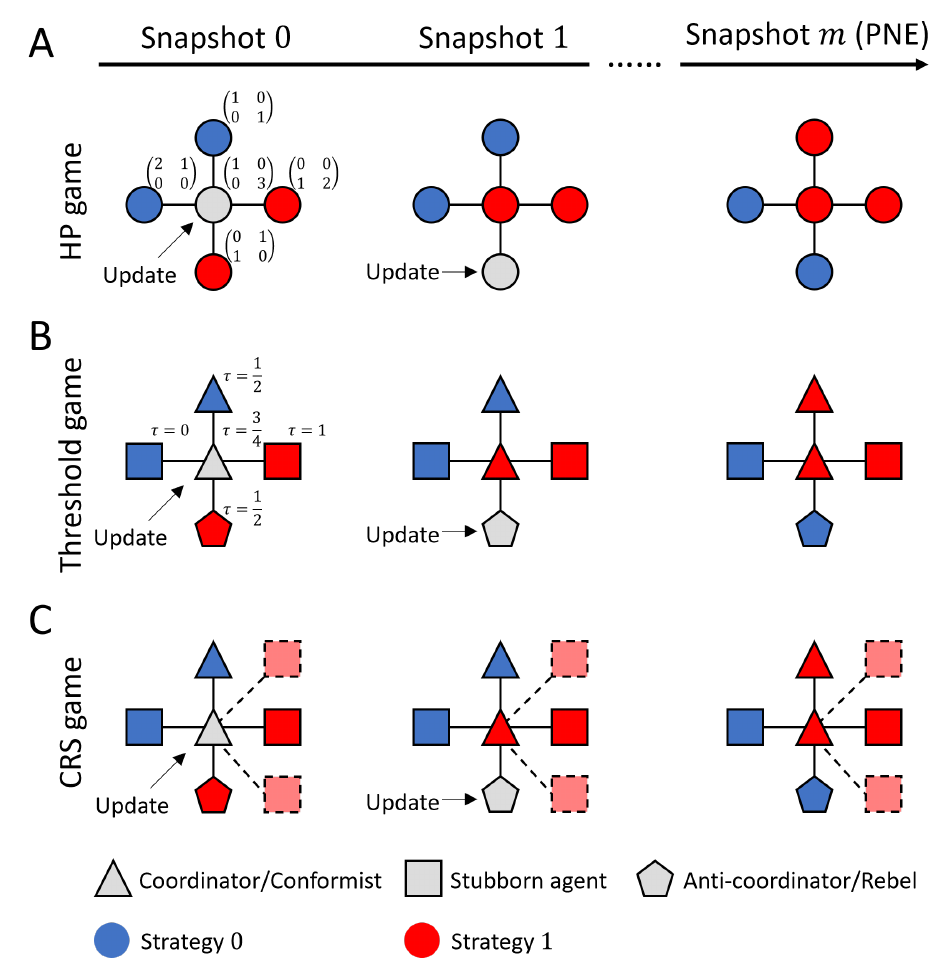}
\caption{\textbf{Equivalence of the dynamics for the HP game, the threshold game, and the CRS game.} (A) In the HP game, each player has a distinct payoff matrix. At each time step, the active player in the snapshots is marked in light gray, and updates their strategy based on the best response dynamics. (B) The HP game can be reformulated as a linear threshold game (with heterogeneous threshold values). Triangles represent coordinators, squares represent stubborn agents, and pentagons represent anti-coordinators. (C) The threshold game can be further transformed into an equivalent CRS game by adding stubborn agents (dashed squares). Triangles represent conformists, squares represent stubborn agents, and pentagons represent rebels.}\label{fig1}
\end{figure}

\section{Pure strategy Nash equilibrium}\label{sec3}

We first provide some conditions for the existence of a PNE in CRS games, and then extend them to HP games. When there is no ambiguity, we also denote the sets of stubborn agents with dominant strategies 0 and 1 in a CRS game by $S_0$ and $S_1$, respectively. 

\begin{theorem}\label{th2}
A network CRS game admits a PNE if one of the following conditions holds: (i) the game does not have conformist–rebel (CR) edges; 
(ii) there exists a strategy $a\in \{0,1\}$ such that each conformist has at least $\frac{1}{2}$ of her neighbors in the set $C \cup S_a$; (iii) there exists a partition of the rebel set $R$, $\{R_0,R_1\}$, such that each rebel in $R_0$ has at least $\frac{1}{2}$ of her neighbors in the set $R_1 \cup S_1$ and each rebel in $R_1$ has at least $\frac{1}{2}$ of her neighbors from the set $R_0 \cup S_0$. 
\end{theorem}

While a formal proof for Theorem \ref{th2} is provided in \textit{Appendix}, section S1B, it is illustrative to present here a sketch of it. For case (i), a CRS game without CR edges can be decomposed into a CS game (a game without rebels) and an RS game (a game without conformists). Both the CS and RS games are exact potential games and therefore possess a PNE. 
For case (ii), when all conformists adopt the strategy $a$, they have no incentive to deviate regardless of the strategies of the rebels. In this case, we can treat conformists as stubborn agents with dominant strategy $a$, and then the CRS game is reduced to an RS game, which has a PNE. For case (iii), we can assign strategy $0$ to rebels in $R_0$ and strategy $1$ to rebels in $R_1$. Then, all rebels have no incentive to deviate regardless of the strategies of the conformists. In this case, we can treat rebels as stubborn agents and the CRS game reduces to a CR game, which has a PNE. 

Using the equivalence between CRS and HP games, he three conditions provided in Theorem \ref{th2} can be directly extended to HP games (the proof is provided in \textit{Appendix}, section S1C).

\begin{corollary}\label{co1}
A network HP game admits a PNE if one of the following conditions holds: (i) the game does not have coordinator–anti-coordinator (CoACo) edges; (ii) each coordinator $i\in Co$ has at least $\tau_i$ proportion of her neighbors from the set $Co \cup S_0$, or each coordinator has at least $1-\tau_i$ proportion of her neighbors from the set $Co \cup S_1$; (iii) there exists a partition of the anti-coordinator set $ACo$, $\{ACo_0, ACo_1\}$, such that each anti-coordinator $i\in ACo_0$ has at least $1-\tau_i$ proportion of her neighbors from the set $ACo_1 \cup S_1$ and each anti-coordinator $j\in ACo_1$ has at least $\tau_j$ proportion of her neighbors from the set $ACo_0 \cup S_0$.
\end{corollary}

Condition (i) in Corollary \ref{co1} can be obtained by applying Theorem 1 and Condition (i) in Theorem 2. \cite{ramazi2016networks} have shown that an HP game with coordinators only or anti-coordinators only admits a PNE. Thus, condition (i) is an extension of their result (by involving stubborn agents). When coordinators and anti-coordinators coexist, conditions (ii) and (iii) imply that a PNE can exist if each coordinator has more coordinating neighbors or each anti-coordinator has more anti-coordinating neighbors. \cite{zhang2018fashion} have shown that a fashion game (a game with conformists and rebels) has a PNE if it satisfies strong conformist homophily or strong rebel homophily, and their results are special cases of conditions (ii) and (iii) in Theorem \ref{th2} and Corollary \ref{co1}.

As for the question of nonexistence of a PNE, to the best of our knowledge there is no general method to determine this for a CRS or HP game on an arbitrarily given network. The next theorem shows that a PNE almost surely does not exist for CRS games on large random sparse networks, where such networks follow a Poisson degree distribution with a small average degree relative to the network size (the proof is provided in \textit{Appendix}, section S1D). 

\begin{theorem}\label{th3}
Consider a network CRS game on a random sparse graph with $n$ nodes and average degree $d$.
If conformists, rebels, stubborn agents with dominant strategy $0$, and stubborn agents with dominant strategy $1$ are randomly distributed on the network, then the game almost surely does not admit a PNE as $n \to \infty$.
\end{theorem}

In the proof of Theorem \ref{th3}, we identify a simple but critical local structure whose existence precludes the game from admitting a PNE. This local structure consists of four players: a degree‑3 rebel with a degree‑1 conformist neighbor and two stubborn neighbors whose dominant strategies are 0 and 1, respectively. In this local structure, the rebel and the conformist continually change their choices: the rebel adopts the strategy opposite to the conformist, but then the conformist switches, prompting the rebel to switch as well, and so on.

We note that many local structures can preclude the existence of PNE. More generally, a pair of nodes $(u,v)$ is said to form a critical local structure if there exist non-negative integers $k_u,k_v\ge 0$ such that node $u \in C$ has $2k_u+1$ neighbors, including one rebel neighbor $v$, $k_u$ stubborn neighbors with dominant strategy 0, and $k_u$ stubborn neighbors with dominant strategy 1; similarly, node $v \in R$ has $2k_v+1$ neighbors, including one conformist neighbor $u$, $k_v$ stubborn neighbors with dominant strategy 0, and $k_v$ stubborn neighbors with dominant strategy 1 (thus, the degrees of $u$ and $v$ must be odd numbers). Clearly, if a CRS game has a pair of nodes that form a critical local structure, then the game admits no PNE. Based on this, Theorem \ref{th3} can be extended to other random networks with diverse degree distributions, such as normal, scale‑free, and exponential distributions, provided that a certain number of odd‑degree nodes persists as the network grows.

\section{Evolutionary dynamics and equilibrium outcomes}\label{sec4}
\subsection{Estimating equilibrium frequencies of strategies in CRS games}

Although Theorem \ref{th2} and Corollary \ref{co1} provide conditions for the existence of a PNE, computing a specific PNE remains difficult, and even testing the existence of PNE for a CRS or HP game is NP-hard \citep{cao2014fashion}. Moreover, Theorem \ref{th3} shows that PNE do not exist for most large networks. Therefore, a central challenge is to understand how the system evolves when no PNE exists. 

We develop a general framework for estimating evolutionary trends and equilibrium strategy frequencies for CRS and HP games on arbitrary networks, without computing or checking the existence of a PNE. Using mean-field and stochastic approximations, we reduce the asynchronous best-response dynamics to a system of ordinary differential equations (ODE) that describes the time evolution of strategy frequencies within the conformist/coordinator and rebel/anti-coordinator subpopulations. We then fully characterize the fixed points and the global dynamics of this ODE system. The framework allows us to evaluate how the network structure and individual heterogeneity shape evolutionary outcomes.

Similarly to the previous subsection, we first consider CRS games and then extend the results to HP games. Given a strategy profile $\boldsymbol{x}$, let $C_0$ and $C_1$ denote the sets of conformists adopting strategies $0$ and $1$, respectively, with sizes $n_{C0}$ and $n_{C1}$. Define $R_0$, $R_1$, $n_{R0}$ and $n_{R1}$ analogously for rebels. Thus, $x_c = \frac{n_{C0}}{n_C}$ and $x_r = \frac{n_{R0}}{n_R}$ represent the frequencies of conformists and rebels adopting strategy $0$. We focus on the time evolution of the equilibrium frequencies of strategies $(x_c, x_r)$ under the asynchronous best response dynamics.

Following \cite{zhang2018fashion} and \cite{PEI2024dynamic}, we use homophily and heterophily indices to measure interactions between different types of players. Given a CRS game, its (global) homophily indices $h_{XX}$ and heterophily indices $h_{XY}$ are defined by
$$
h_{XX}= \frac{2K_{XX}}{2K_{XX} + \sum_{Z \in \{C, R, S_0, S_1\} \setminus \{X\}} K_{XZ}},
$$
$$
h_{XY}= \frac{K_{XY}}{2K_{XX} + \sum_{Z \in \{C, R, S_0, S_1\} \setminus \{X\}}K_{XZ}},
$$
where $X, Y\in \{C, R, S_0, S_1\}$ and $K_{XY}$ denotes the number of edges between players of types $X$ and $Y$. Intuitively, $h_{XY}$ measures the probability that a randomly chosen neighbor of a type $X$ player is type $Y$. 

Let $v_i(\boldsymbol{x}) = \frac{u_i(\boldsymbol{x})}{n_i}$ be the normalized payoff for player $i$. The average normalized payoffs for players of types $C_0$, $C_1$, $R_0$, and $R_1$ can be approximated as
\begin{equation}
\begin{cases}
v_{C0} \approx h_{CC}x_c+h_{CR}x_r+h_{CS_{0}} \\
v_{C1} \approx h_{CC}(1-x_c) + h_{CR}(1-x_r) + h_{CS_{1}} \\
v_{R0} \approx h_{RR}(1-x_r) + h_{RC}(1-x_c) + h_{RS_{1}} \\
v_{R1} \approx h_{RR} x_r + h_{RC} x_c + h_{RS_{0}}
\end{cases}.
\label{eq1}
\end{equation}
In the derivation of Eq.~(\ref{eq1}), we apply two mean-field approximations. First, we approximate the distribution of different types of neighbors among a conformist (or a rebel) by the homophily index $h_{CX}$ (or $h_{RX}$). Second, we approximate the distribution of neighbors' strategies among a player by the global strategy distributions $x_c$ and $x_r$. A detailed derivation of Eq.~(\ref{eq1}) is given in \textit{Appendix}, section S1E.

Next, by applying the stochastic approximation method \citep[see e.g.][]{sandholm2010population,roth2013stochastic,PEI2024dynamic}, the time evolution of $(x_c, x_r)$ under the asynchronous best response dynamics with stochastic deviating orders and action switchings
can be described by the following ODE system in the limit of $n\to \infty$ and continuous time:
\begin{equation}
\begin{cases}
\frac{dx_c}{dt} = (1 - x_c)\phi(v_{C1} - v_{C0}) - x_c\phi(v_{C0} - v_{C1}) \\
\frac{dx_r}{dt} = (1 - x_r)\phi(v_{R1} - v_{R0}) - x_r\phi(v_{R0} - v_{R1}) 
\end{cases},
\label{eq2}
\end{equation}
where $\phi(\cdot)$ is the strategy switching function with $\phi(v)> 0$ when $v<0$ and $\phi(v)= 0$ when $v\geq 0$. The complete derivation of ODE (\ref{eq2}) is provided in \textit{Appendix}, section S1F, and the phase portrait of ODE (\ref{eq2}) is shown in Figure \ref{fig2}.

\begin{figure}[H]
\centering
\includegraphics[width=1\linewidth]{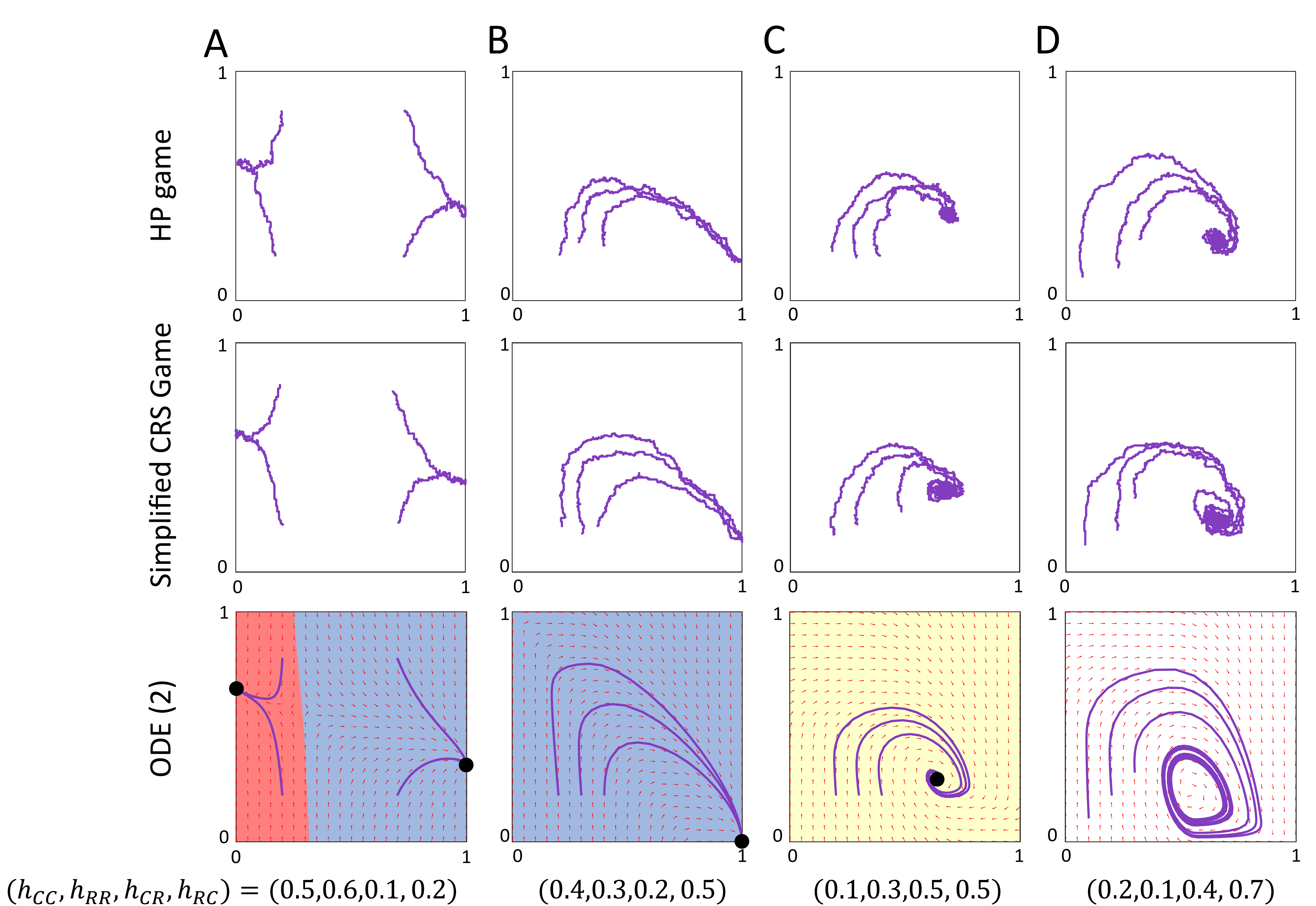}
\caption{\textbf{Comparisons between trajectories of the HP game (top), its simplified CRS game (middle), and the corresponding ODE (\ref{eq2}) (bottom)}. Horizontal and vertical axis correspond to the $x_c$ and $x_r$, respectively. Purple curves denote the trajectories. Solid points in the bottom row denote the stable fixed points of ODE (\ref{eq2}), and red, blue, and yellow regions denote the basins of attraction of the fixed points $(0,*)$, $(1,*)$, and $(\bar{x}_c, \bar{x}_r)$, respectively. In numerical simulations (top and middle rows), the network size is taken as $n=1000$ and the average degree is $16$, with $n_C=n_R=n_{S_0}=n_{S_1}=250$ in simplified CRS games. In addition, we take $h_{CS_0}=0.3$ and $h_{CS_1}=h_{RS_0}=h_{RS_1}=0.1$ in all the subfigures, and adjust $h_{CC}, h_{RR}, h_{CR}$, and $h_{RC}$ according to Theorem \ref{th4}. Details for network generation and agent-based simulation can be found in \textit{Appendix}, section S2C.}\label{fig2}
\end{figure}

We then characterize the dynamic properties of ODE (\ref{eq2}).  Table \ref{tab1} summarizes the existence and stability conditions for all possible fixed points. In addition, the global dynamic behavior of ODE (\ref{eq2}) is presented in Theorem 4 (the proof is provided in \textit{Appendix}, section S1H).

\begin{table}[h]
\centering
\small
\caption{\textbf{The existence and stability conditions of fixed points of ODE (\ref{eq2})}. ODE (\ref{eq2}) has at most one interior fixed point $(\bar{x}_c, \bar{x}_r) = \left( \frac{x_c^* h_{RR} - x_r^* h_{CR}}{h_{CC}h_{RR} - h_{CR}h_{RC}}, \frac{x_r^* h_{CC} - x_c^* h_{RC}}{h_{CC}h_{RR} - h_{CR}h_{RC}} \right)$, where $x_c^* = \frac{1}{2}-h_{CS_0}$ and $x_r^* = \frac{1}{2}-h_{RS_0}$.}\label{tab1}
\setlength{\tabcolsep}{18pt}
\renewcommand{\arraystretch}{1}
\begin{tabular}{p{2cm}p{6cm}p{5.5cm}}
\toprule
\textbf{Fixed points} & \textbf{Existence conditions} & \textbf{Stability conditions} \\
\midrule
$(0,0)$ & $h_{CR} \geq -h_{CC} + h_{CS_0} - h_{CS_1}$, $h_{RC} \leq -h_{RR} + h_{RS_0} - h_{RS_1}$ & Strict inequalities hold \\
$(1,1)$ & $h_{CR} \geq -h_{CC} - h_{CS_0} + h_{CS_1}$, $h_{RC} \leq -h_{RR} - h_{RS_0} + h_{RS_1}$ & Strict inequalities hold \\
$(0,1)$ & $h_{CR} \leq h_{CC} - h_{CS_0} + h_{CS_1}$, $h_{RC} \geq h_{RR} + h_{RS_0} - h_{RS_1}$ & Strict inequalities hold \\
$(1,0)$ & $h_{CR} \leq h_{CC} + h_{CS_0} - h_{CS_1}$, $h_{RC} \geq h_{RR} - h_{RS_0} + h_{RS_1}$ & Strict inequalities hold \\
$(0, \frac{x_r^*}{h_{RR}})$ & $-h_{RR} + h_{RS_0} - h_{RS_1} < h_{RC} < h_{RR} + h_{RS_0} - h_{RS_1}$, $\frac{1-2h_{RS_0}}{h_{RR}} \leq \frac{1-2h_{CS_0}}{h_{CR}}$ & Strict inequalities hold \\
$(\frac{x_c^*}{h_{CC}}, 0)$ & $-h_{CC} + h_{CS_0} - h_{CS_1} < h_{CR} < h_{CC} + h_{CS_0} - h_{CS_1}$, $\frac{1-2h_{RS_0}}{h_{RC}} \leq \frac{1-2h_{CS_0}}{h_{CC}}$ & Always unstable \\
$(1, \frac{x_r^*-h_{RC}}{h_{RR}})$ & $-h_{RR} - h_{RS_0} + h_{RS_1} < h_{RC} < h_{RR} - h_{RS_0} + h_{RS_1}$, $\frac{1-2h_{RS_1}}{h_{RR}} \leq \frac{1-2h_{CS_1}}{h_{CR}}$ & Strict inequalities hold \\
$(\frac{x_c^*-h_{CR}}{h_{CC}}, 1)$ & $-h_{CC} - h_{CS_0} + h_{CS_1} < h_{CR} < h_{CC} - h_{CS_0} + h_{CS_1}$, $\frac{1-2h_{RS_1}}{h_{RC}} \leq \frac{1-2h_{CS_1}}{h_{CC}}$ & Always unstable \\
$(\bar{x}_c, \bar{x}_r)$ & $h_{CC} h_{RR} \neq h_{CR} h_{RC}$, $0 < \bar{x}_c < 1$, $0 < \bar{x}_r < 1$ & $h_{CR} h_{RC}>h_{CC} h_{RR}$, $h_{CC}<h_{RR}$ \\
\bottomrule
\end{tabular}
\end{table}

\begin{theorem}\label{th4}
(i) If $h_{CR} h_{RC} < h_{CC} h_{RR}$, then ODE (\ref{eq2}) has stable boundary fixed points. (ii) If $h_{CR} h_{RC} > h_{CC} h_{RR}$, $h_{CC} < h_{RR}$, $0 < \bar{x}_c < 1$ and $0 < \bar{x}_r < 1$, then ODE (\ref{eq2}) has a stable interior fixed point.
(iii) If $h_{CR} h_{RC} > h_{CC} h_{RR}$, $h_{CC} > h_{RR}$, $0 < \bar{x}_c < 1$ and $0 < \bar{x}_r < 1$, then a limit cycle emerges.
\end{theorem}

It is interesting to note that while Theorem \ref{th2} establishes the existence of a PNE in the absence of CR edges, Theorem \ref{th4} provides a complementary analysis in the presence of CR edges. Since $h_{CC}h_{RR} > h_{CR}h_{RC}$ implies $K_{CR} < 2\sqrt{K_{CC}K_{RR}}$, Theorem \ref{th4} and Table \ref{tab1} suggest that trajectories of ODE (\ref{eq2}) converge to boundary fixed points where all conformists coordinate on the same strategy when there are few CR edges. In contrast, trajectories of ODE (\ref{eq2}) may not converge and periodic oscillations are likely to emerge when there are more CR edges. In this regard, 
Corollary \ref{co2}, which follows directly from Theorem \ref{th4}, further shows that a larger number of CR edges is a necessary condition for the existence of limit cycles. 

\begin{corollary}\label{co2}
If ODE (\ref{eq2}) has limit cycles, then $K_{CR} > 2\sqrt{K_{CC}K_{RR}}$. 
\end{corollary}

\subsection{Estimating equilibrium frequencies of strategies in HP games}

Since the transformation from any HP game to its equivalent CRS game does not alter the number of CC, RR, and CR edges (Theorem \ref{th1}), Theorem \ref{th4} and Corollary \ref{co2} can be extended to HP games. Corollary \ref{co3} shows that in an HP game, a small number of CoACo edges drives the trajectories of the corresponding ODE system of its equivalent CRS game toward boundary fixed points, whereas a large number of such edges facilitates the emergence of periodic fluctuations.

\begin{corollary}\label{co3}
(i) If $K_{CoACo} < 2\sqrt{K_{CoCo}K_{ACoACo}}$,  then the corresponding ODE system of the network HP game has a stable boundary fixed point. (ii) If the corresponding ODE system of the network HP game has limit cycles, then $K_{CoACo} > 2\sqrt{K_{CoCo}K_{ACoACo}}$. 
\end{corollary}

From Theorem \ref{th1}, an HP game usually has fewer stubborn agents than its equivalent CRS game, which leads to very different homophily/heterophily indices. As a result, the approximation method introduced in the above subsection cannot be directly applied to characterize the effect of network structure on the strategy evolution in an HP game.

Here we propose a simplified method that transforms the HP game into a CRS game without altering the original network. For an HP game $G=(N,E,A,U)$, we define its simplified CRS game as $G'=(C, R, S, E, A, V)$, where $C=Co$, $R=ACo$, and $S$ and $E$ are the same as the HP game. In other words, the simplified CRS game of an HP game is obtained by replacing coordinators and anti-coordinators with conformists and rebels, respectively. Then we can calculate homophily/heterophily indices for this simplified CRS game and apply the results in Theorem \ref{th4} and Table \ref{tab1} to predict evolutionary trends and estimate the equilibrium frequencies of strategies in the original HP game. Of course, it has to be kept in mind that as this simplified CRS game has fewer CS and RS edges, homophily indices for conformists and rebels may be overestimated. 

\subsection{Error analysis}

To test the robustness of the simplified method and the deterministic approximation method, we compare the numerical trajectories and fixed points of the HP game, its simplified CRS game, and the corresponding ODE (\ref{eq2}) for four combinations of homophily and heterophily indices (see \textit{Appendix}, sections S2A-S2C for details). In each combination, we plot the trajectories for 3 or 4 initial points (see Figure \ref{fig2}). In the HP game, the payoff matrix for each coordinator or anti-coordinator is generated randomly and independently. Numerical simulations show that the short-run behavior of the HP game could be nicely approximated by the trajectories of the simplified CRS game and ODE (\ref{eq2}). In general, the trajectories of the HP game and the simplified CRS game converge fairly accurately to the fixed points predicted by ODE (\ref{eq2}), and the evolutionary trends of these trajectories match the phase portrait of ODE (\ref{eq2}). Specifically, trajectories converge to PNE where all coordinators adopt the same strategy when there are few CoACo edges (i.e., low $h_{CR}$ and/or $h_{RC}$ in the simplified CRS game), and a greater number of CoACo edges facilitate the emergence of periodic fluctuations. 

Since the derivation of ODE (\ref{eq2}) relies on mean-field approximations and the stochastic approximation, its stable fixed points need not coincide with PNE of the CRS game. The simplified method, which replaces coordinators and anti-coordinators with conformists and coordinators, also introduces errors. We further quantitatively evaluate these discrepancies through two comparisons: HP games versus their simplified CRS games, and CRS games versus ODE (\ref{eq2}). In the HP \& CRS comparison, we fix the mean threshold of coordinators and anti-coordinators at $\frac{1}{2}$ and vary its variance. Simulation results across diverse networks show that estimation error increases with the threshold variance. In the CRS \& ODE comparison, error grows with the number of stable fixed points and  becomes prominent near the ODE’s bifurcation points. A detailed error analysis is provided in \textit{Appendix}, section S2E.

\section{Limited information}\label{sec5}
In order to further generalize our results, we now consider the case that, due to restrictions on  information, people may only be able to observe the actions of some neighbors and choose the best response to the distribution of actions in the sample \citep{oyama2015sampling}. To this end, we define the probability that a player $i$ observes the strategy of a neighboring player $j$ by $q_{ij}$, and the set of $i$'s neighbors by $N_i$. If $q_{ij}=1$ for all $i\in N$ and all $j\in N_i$, then players always observe the strategies of all neighbors, corresponding to a game with complete information. If $0<q_{ij}<1$ for all $i\in N$ and all $j\in N_i$, then players typically observe only the strategies of some neighbors, corresponding to a game with limited information. We note that the limited information setting considered here differs from the classic incomplete information framework in economic theory, wherein such scenarios are typically reformulated as Bayesian games and agents can infer neighbors’ opinion distributions from prior beliefs about system states and past experience.

Suppose that at time $t$, player $i$ observes the strategies of the players in a set $\widetilde{N}_i(G,t)\subseteq N_i$. Then, her (myopic) best response strategy in time $t+1$ can be denoted by $x_i(t+1)=\operatorname{argmax}_{x_i \in A}\left\{\tilde{u}_i\left(x_i, \boldsymbol{x}_{-i}(t)\right)\right\}$, 
where
\begin{equation}\label{eq3}
\tilde{u}_i\big(x_i,\boldsymbol{x}_{-i}(t)\big)=\left\{\begin{array}{cl}
a_i \tilde{n}_i^0(t) + b_i \tilde{n}_i^1(t), & \text { if } x_i = 0 \\
c_i \tilde{n}_i^0(t) + d_i \tilde{n}_i^1(t), & \text { if } x_i = 1
\end{array}\right. 
\end{equation}
is the expected payoff of player $i$ using strategy $x_i$ based on the information obtained in time $t$ and $\tilde{n}_i^a(t) := |\{ j\in \widetilde{N}_i(t): x_j = a\}|$. In addition, if $\widetilde{N}_i(G,t)=\emptyset$ (i.e., player $i$ does not observe any information), then the utility satisfies $\tilde{u}_i(x_i,\boldsymbol{x}_{-i}(t))=0$ for $ \forall x_i\in A$. In this case, player $i$ will not change his/her current strategy. We define a strategy profile $\boldsymbol{x}^*=\left(x_1^*, \ldots, x_n^*\right) \in A^N$ as a limited information PNE (L-PNE) if $\tilde{u}_i\left(\boldsymbol{x}^*\right) \geq \tilde{u}_i\left(x_i, \boldsymbol{x}_{-i}{ }^*\right)$ for all $i \in N$, $\widetilde{N}_i \subseteq$ $N_i$, and $x_i \in A$, where $\tilde{u}_i(\boldsymbol{x})$ is defined in Eq.(\ref{eq3}).

In this section, we assume that the HP game includes all three player types: coordinators, anti-coordinators, and stubborn agents. We next show that limited information can significantly affect the equilibrium structure of an HP game. At an L-PNE, each coordinator adopts the same strategy as all its neighbors, implying that all its neighbors share the same strategy; each anti-coordinator chooses a strategy different from that of every neighbor. Clearly, an L-PNE constitutes a refinement of the PNE under complete information. The next theorem provides necessary and sufficient conditions for the existence of an L-PNE.

\begin{theorem}\label{th5}
A network HP game admits an L-PNE if and only if the following three conditions hold: (i) the game does not have CoACo edges; (ii) for each subnetwork composed of coordinators, all stubborn neighbors adopt the same strategy; (iii) for each subnetwork composed of anti-coordinators ${ACo}^k$, there exists a partition, $\{ACo_0^k, ACo_1^k\}$, such that no two players within the same set are connected by an edge and all stubborn neighbors of players in set $ACo_{0}^k$ (or $ACo_{1}^k$) adopt strategy $1$ (or $0$).
\end{theorem}

The proof of Theorem \ref{th5} is provided in \textit{Appendix}, section S1I. From the proof, if an HP game admits an L-PNE, then it must be unique and is globally stable under the best response dynamics, i.e., the dynamics will converge to the unique L-PNE for any initial strategy distribution. The next corollary shows that if an HP game does not have an L-PNE, then the best response dynamics, which is a finite-state Markov process, will converge to a unique stationary distribution over time (the proof is provided in \textit{Appendix}, section S1J).

\begin{corollary}\label{co4}
In the case of limited information, the asynchronous best response dynamics either converges to an L-PNE (if it exists) or to a stationary distribution. 
\end{corollary}

Corollary \ref{co4} implies that under limited information, the effect of network structure on the evolutionary outcome can be well evaluated through agent-based simulations (see Figure \ref{fig3}). In addition, the result of Corollary \ref{co4} also holds for the synchronous best response dynamics (see \textit{Appendix}, section S1J). Figure \ref{fig3} shows the influences of information transparency (complete information or $q_{ij}=0.1, 0.5$) and update mode (asynchronous or synchronous) on the evolutionary outcome of CRS games. When the CRS game has a unique L-PNE, both asynchronous and synchronous best response dynamics converge to it under limited information (Figure \ref{fig3}A). In contrast, when the game does not have an L-PNE, the best response dynamics under limited information generally converges to a stationary distribution where the corresponding strategy distribution is close to a stable fixed point of ODE (\ref{eq2}) (Figure \ref{fig3}B-C). Finally, information transparency ($q_{ij}=0.1$ or $0.5$) and update mode seem to have little impact on the mean value of the stationary strategy distribution when the game does not process a PNE (Figure \ref{fig3}C-D).

\begin{figure}[H]
\centering
\includegraphics[width=0.75\linewidth]{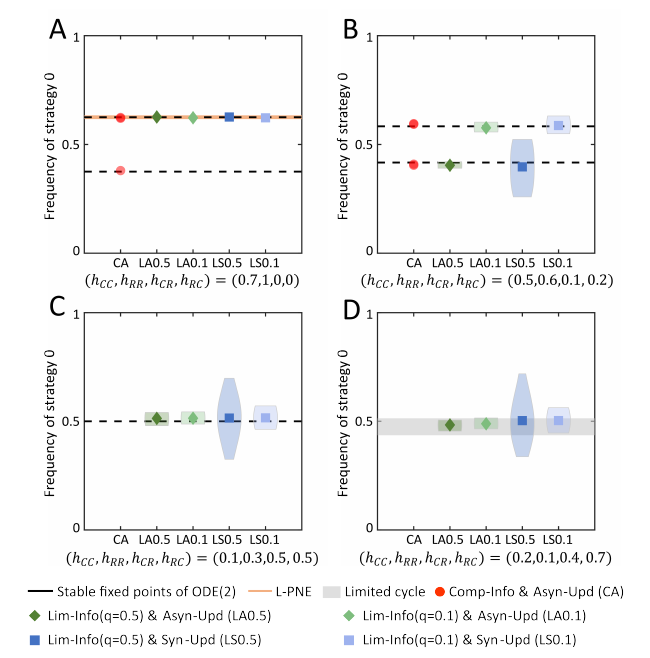}
\caption{\textbf{The influences of information transparency and update mode on the evolutionary outcome.} Similarly as Figure \ref{fig2}, the network size of CRS games is taken as $n=1000$ with $n_C=n_R=n_{S_0}=n_{S_1}=250$ and the average degree is $16$. In addition, we take $h_{CS_0}=0.3$ and $h_{CS_1}=h_{RS_0}=h_{RS_1}=0.1$ in all the subfigures, and we adjust $h_{CC}, h_{RR}, h_{CR}$, and $h_{RC}$. Orange solid line denotes the frequency of strategy $0$ at the L-PNE, black dashed lines denote the frequency of strategy $0$ at stable fixed points of ODE (\ref{eq2}), and gray shaded area denotes the range of strategy frequency at the limit cycle. Comp-Info, Lim-Info, Asyn-Upd, and Syn-Upd represent complete information, limited information (with $q_{ij}=0.1$ or $0.5$), asynchronous update, and synchronous update, respectively. Red dots denote the fixed points of the asynchronous best response dynamics under complete information, which correspond to PNEs (we note that PNE does not exist in C and D). Diamonds and squares denote the mean value of the stationary strategy distribution (SSD) under limited information. In numerical simulations, we ran $10^6$ time steps to ensure convergence to PNE or SSD. Under complete information, PNEs exist when the homogeneity index is relatively high (A and B), and convergence to distinct equilibria occurs for different initial strategy distributions. Under limited information, both synchronous and asynchronous updates yield a unique stationary strategy distribution regardless of the initial strategy distribution, with the distribution plotted using the strategy proportions from the final $10^4$ time steps.}\label{fig3}
\end{figure}

\section{Application: Prisoner’s Dilemma game on networks with preference heterogeneous players}\label{sec6}
As an application, and to show the validity of our framework, we consider a Prisoner’s Dilemma game on networks where the utility function of each player comprises three components: material payoff, altruistic preference, and peer influence. In this game, each player chooses between cooperation ($C$) or defection ($D$). A cooperator pays a cost $1$ for each neighbor, and each neighbor receives a benefit $b>1$. A defector neither incurs costs nor confers benefits. Beyond material payoff, altruistic preferences grant a player $i$ a utility gain equal to $\alpha_i>0$ times the material payoff of each neighbor. Additionally, peer influence contributes a change in utility equal to $\beta>0$ times the weight of the relationship when two connected players adopt the same strategy. Thus, the utility matrix of player $i$ versus player $j$ can be expressed as follows:
\begin{equation}
\bordermatrix{
   & C & D \cr
C & b-1+\alpha_i(b-1)+\beta w_{ij}  & \alpha_i b -1 \cr
D & b-\alpha_i & \beta w_{ij} \cr
}, 
\label{eq4}
\end{equation}
where $w_{ij}$ is the weight of the relationship from $i$ to $j$. In addition, $w_{ij} > 0$ means that the neighbor is a friend and $w_{ij} < 0$ means that the neighbor is an enemy.

Following our approach, we transform this game into a simplified CRS game. Since players have different payoff matrices versus different neighbors, we cannot use the linear-threshold model (which is only based on the number of neighbors with different strategies) to represent their best response strategies. Therefore, we classify players as conformists, rebels, and stubborn agents according to their optimal strategy when all neighbors adopt the same strategy. Specifically, we refer to a player $i$ as a conformist if her optimal strategy is $C$ when all neighbors adopt $C$ and is $D$ when all neighbors adopt $D$. This requires $-\beta \omega_i<(\alpha b - 1)n_i < \beta \omega_i$ with $\omega_i=\sum_{j \in N_i} w_{ij}$. In contrast, a player is referred to as a rebel if her optimal strategy is $D$ when all neighbors adopt $C$ and is $C$ when all neighbors adopt $D$. This requires $\beta \omega_i<(\alpha b - 1)n_i < -\beta \omega_i$. Finally, a player is referred to as a stubborn agent with strategy $C$ (cooperator for short) if $(\alpha b - 1)n_i > |\beta \omega_i|$ and a stubborn agent with strategy $D$  (defector for short) if  $(\alpha b - 1)n_i <- |\beta \omega_i|$. In summary, a cooperator has a larger $\alpha$ and a defector has a smaller $\alpha$. In addition, a conformist (or rebel) has an intermediate range of $\alpha$ (i.e., $\alpha b \approx 1$) and a positive (or negative) $\omega_i$.

We now evaluate the accuracy of the simplified method and the deterministic approximation method by conducting agent-based simulations based on real social networks of 13 schools collected by \cite{Miguel2023}. In their data set, students report their levels of prosociality (see also \cite{Vasco2025} for further details) and positive or negative relationships with their neighbors. In our simulations, the level of prosociality of a player is mapped to the altruism parameter $\alpha_i$, the peer influence parameter is fixed at $\beta = 1/4$ for all players, and the benefit of cooperation is set to $b = 2$. In this manner, we find that three out of 13 networks involve both conformists and rebels. For these three networks, the differences between the cooperation rates in stable fixed points of the best response dynamics based on payoff matrix (\ref{eq4}) and its simplified game are less than $5\%$, and the differences between the simplified game and its deterministic approximation are less than $1 \%$. These results demonstrate that the two methods can effectively capture the joint effect of preference heterogeneity and network structure on group cooperation. In addition, the deterministic approximation method provides a straightforward prediction that the cooperation rate in the three networks is approximately equal to the sum of the proportions of conformists and cooperators (see \textit{Appendix}, section S3A).

Within this network Prisoner’s Dilemma game, stubborn defectors never cooperate, and the theoretical upper bound on the cooperation rate equals the combined proportion of conformists, rebels, and stubborn cooperators. Consequently, cooperation can be enhanced through two mechanisms. First, the upper bound itself can be raised by reducing the prevalence of stubborn defectors. Second, this upper bound (corresponding to the fixed point $(1,1)$ in Table \ref{tab1}) becomes stable when conformists connect more to stubborn cooperators than to stubborn defectors, whereas rebels display the opposite connectivity pattern (i.e., $h_{CS_0}>h_{CS_1}$ and $h_{RS_0}<h_{RS_1}$). Collectively, these findings offer pathways for promoting cooperation among preference heterogeneous players.

\begin{figure}[H]
\centering
\includegraphics[width=0.75\linewidth]{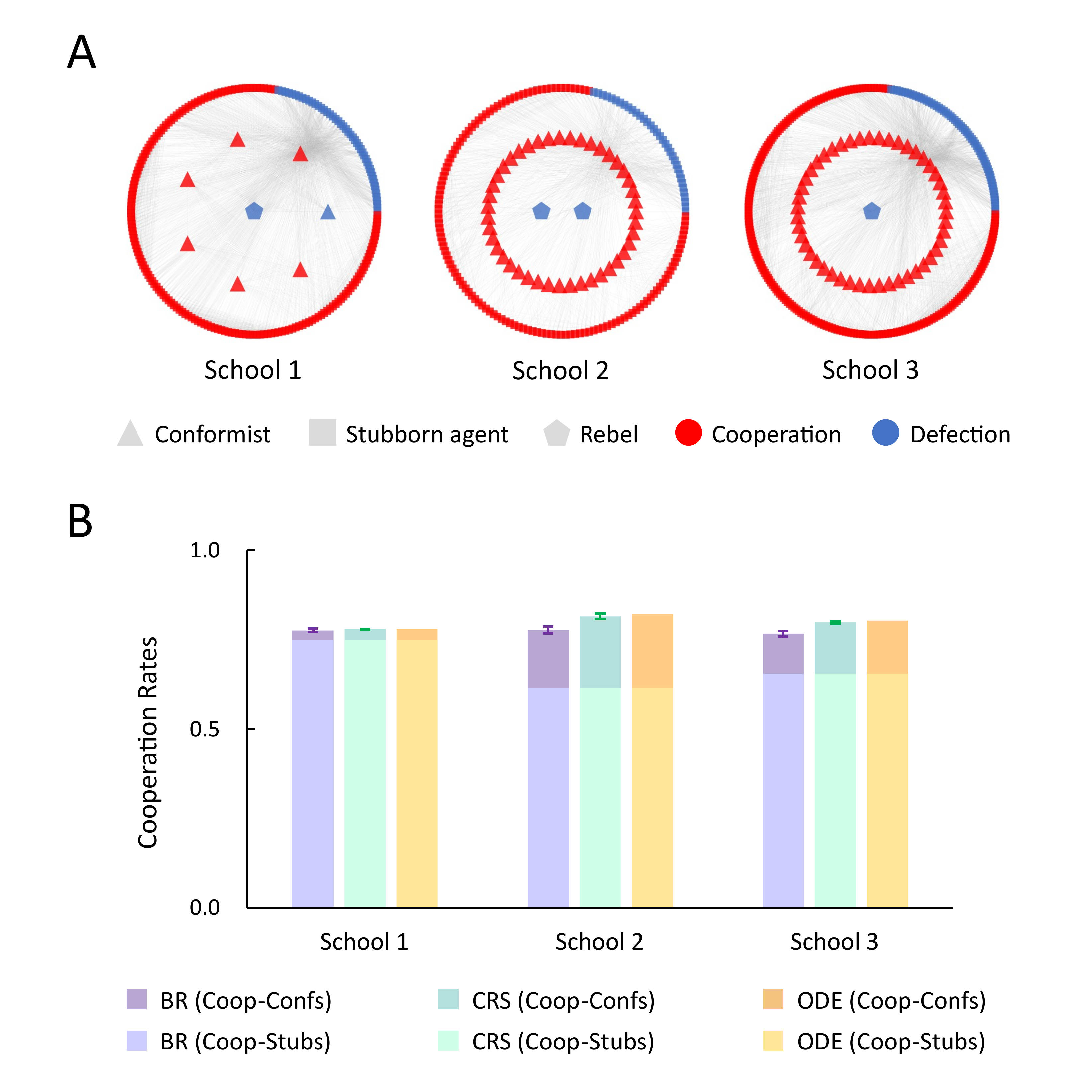}
\caption{\textbf{Prisoner’s Dilemma game on networks with preference heterogeneous players.} The top panel provide illustrations for PNE in the 3 school networks. Stubborn agents, conformists, and rebels are distributed on the outer ring, middle ring, and inner ring, respectively, and their strategies are marked by red and blue. Coop-Confs and Coop-Stubs represent cooperative conformists (dark color) and stubborn cooperators (light color), respectively. The bottom panel compares the cooperation rates at stable fixed points of the best response dynamics based on payoff matrix (\ref{eq4}), its simplified CRS game, and the corresponding ODE (\ref{eq2}) in the 3 networks. 
}\label{fig4}
\end{figure}

\section{Discussion}\label{sec7}

\subsection{Concluding remarks}
Our paper advances the understanding of strategic interactions in network games with general heterogeneous players by establishing three principal contributions. First, we derive sufficient conditions for the existence of a PNE, providing a set of practically verifiable criteria for equilibrium analysis. Second, we introduce a deterministic approximation framework that predicts strategy evolution on arbitrary networks. Finally, we prove that, under limited information, the best response dynamics converges either to a PNE or to a unique stationary strategy distribution. Collectively, these results offer a unified analytical framework for examining equilibrium characterization, evolutionary trajectories, and dynamic predictability in network games with heterogeneous players.

\subsection{Extensions: CRS games with $n$ strategies}\label{sec7.2}


Although the present work focuses on two-strategy settings, several findings can extend to CRS games with $n$ strategies. Theorem \ref{th2}(i) holds for $n$-strategy CRS games, because a PNE is guaranteed in such games without CR edges and the potential function used in its proof remains valid for the $n$-strategy case. Theorem \ref{th2}(ii) can also extend to $n$ strategies, because there exists a strategy in which all conformists have no incentive to deviate, reducing the game to an $n$-strategy RS game, where the potential function of its proof guaranties a PNE.

\begin{corollary}\label{co5}
An $n$-strategy network CRS game admits a PNE if one of the following conditions holds: (i) the game does not have CR edges; (ii) there exists a strategy $a\in A$ such that each conformist has at least $\frac{1}{2}$ of her neighbors in the set $C \cup S_a$.
\end{corollary}

In addition, Theorem \ref{th3} also holds for the $n$-strategy cases. In fact, we consider a local structure that consists of $n+2$ players: a degree-$(n+1)$ rebel with one degree-1 conformist neighbor and $n$ stubborn neighbors with distinct dominant strategies. If this local structure exists in a CRS game, then the game does not have a PNE. It is easy to verify that as the size of a CRS game increases, the probability that no such local structure exists tends to zero in a sparse random graph.

\begin{corollary}\label{co6}
Consider an $n$-strategy network CRS game on a random sparse graph with $n$ nodes and average degree $d$. If conformists, rebels, stubborn agents with dominant strategy $0$, and stubborn agents with dominant strategy $1$ are randomly distributed on the network, then the game almost surely does not admit a PNE as $n \to \infty$.
\end{corollary}

Finally, Theorem \ref{th5} can also extend to the $n$-strategy case, where conditions (i) and (ii) are identical to those in the $2$-strategy case, and condition (iii) is revised to state that the maximum degree of anti-coordinator subnetworks is at most $n-1$. This is because anti-coordinators adopt distinct strategies from all their neighbors at an L-PNE, and the existence condition thus follows from graph coloring arguments.

\begin{corollary}\label{co7}
A $n$-strategy network HP game admits an L-PNE if the following three conditions hold: (i) the game does not have CoACo edges; (ii) for each subnetwork composed of coordinators, all stubborn neighbors adopt the same strategy; (iii) for each subnetwork composed of anti-coordinators ${ACo}^k$, the maximum degree is at most $n-1$.
\end{corollary}

These extensions indicate that the analytical framework developed in this paper offers robust insights applicable to a broader class of network games beyond binary strategy spaces.

\subsection{Limitations}
However, several limitations remain. Regarding PNE existence, we establish sufficient conditions for network games with heterogeneous players over arbitrary networks, whereas characterizing necessary and sufficient conditions for general networks remains challenging and is left for future investigations on specific network topologies. 
For the short-run behavior of the evolutionary process, we develop the simplified method and the deterministic approximation method to predict the strategy evolution in an HP game. While both methods entail certain inaccuracies, more precise characterizations await further investigation. 
For long-run outcomes, we derive conditions ensuring the existence and uniqueness of the stationary distribution under limited information. Numerical simulations indicate that the resulting strategy distribution approximates a stable fixed point of ODE (\ref{eq2}). However, when ODE (\ref{eq2}) admits multiple stable fixed points, the precise mechanism by which information shapes equilibrium strategy frequencies remains unresolved.
Finally, we note that an HP game cannot be reformulated as a CRS game when there are more than two strategies. Thus, developing an analytical framework for $n$-strategy HP games would be a challenging question for future studies. 
All in all, we believe that our approach will prove very useful to the community of researchers in network games and their applications in economics, sociology, or biology.

\section*{Acknowledgments}
Wenjie Cao acknowledges support from the National Natural Science Foundation of China (No.724B2006) and the China Scholarship Council (No.202406040147). Angel Sánchez acknowledges support from grant PID2022-141802NB-I00 (BASIC) funded by MCIN/AEI/10.13039/501100011033 and by ‘ERDF way of making Europe’, and also from grant MapCDPerNets---Programa Fundamentos de la Fundaci\'on BBVA 2022. Boyu Zhang acknowledges support from the National Natural Science Foundation of China (No.72131003 and No.72573024) and the Beijing Natural Science Foundation (No.Z220001).
\bibliographystyle{elsarticle-harv}
\bibliography{refs}

\newpage

\appendix

\renewcommand{\thefigure}{S\arabic{figure}}
\renewcommand{\thetable}{S\arabic{table}}
\renewcommand{\theequation}{S\arabic{equation}}
\renewcommand{\thesection}{S\arabic{section}}
\renewcommand{\thesubsection}{\Alph{subsection}}

\setcounter{section}{0}
\section{Theoretical analysis}
\subsection{Proof of Theorem 1}

\setcounter{theorem}{0}
\setcounter{corollary}{0}
\setcounter{figure}{0}

\begin{theorem}
For any HP game, we can construct an equivalent CRS game by adding stubborn neighbors to coordinators and anti-coordinators.
\end{theorem}

\begin{proof}
Let $G = (N, E, A, U)$ denote an arbitrary HP game. For each player $i$ with $\delta_i \neq 0$ and threshold $\tau_i \in (0,1)$, we append a set of new stubborn neighbors to $i$ following the two rules below:

(i) If $\tau_i n_i$ is an integer (i.e., $\tau_i n_i \in \mathbb{N}$), we add $n_i - 2\tau_i n_i$ stubborn neighbors adopting strategy $0$ for players with $\tau_i \in (0, 1/2]$, and $2\tau_i n_i - n_i$ stubborn neighbors adopting strategy $1$ for players with $\tau_i \in (1/2, 1)$.

(ii) If $\tau_i n_i$ is non-integer (i.e., $\tau_i n_i \notin \mathbb{N}$), let $K_i$ be the unique integer satisfying $K_i \in \mathbb{N} \cap [0, n_i)$ and $K_i < \tau_i n_i < K_i + 1$. We add $n_i - 2K_i - 1$ stubborn neighbors with strategy $0$ for players with $K_i \leq (n_i - 1)/2$, and $2K_i + 1 - n_i$ stubborn neighbors with strategy $1$ for players with $K_i > (n_i - 1)/2$.

By applying this procedure to all players in the network $(N, E)$, we obtain a new network $(N', E')$. For any given strategy profile, the best response of player $i$ in the original HP game $G$ is equivalent to adopting the majority strategy (for coordinators) or the minority strategy (for anti-coordinators) among its neighbors in the new game $G'=(N', E', A, V)$. 
\end{proof}

\subsection{Proof of Theorem 2}
\begin{theorem}\label{th2}
A network CRS game admits a PNE if one of the following conditions holds: (i) the game does not have conformist–rebel (CR) edges; 
(ii) there exists a strategy $a\in \{0,1\}$ such that each conformist has at least $\frac{1}{2}$ of her neighbors in the set $C \cup S_a$; (iii) there exists a partition of the rebel set $R$, $\{R_0,R_1\}$, such that each rebel in $R_0$ has at least $\frac{1}{2}$ of her neighbors in the set $R_1 \cup S_1$ and each rebel in $R_1$ has at least $\frac{1}{2}$ of her neighbors from the set $R_0 \cup S_0$. 
\end{theorem}

\begin{proof}
We verify the existence of a PNE under each condition. 

For condition (i), the game does not have CR edges, then conformists are connected only to other conformists or to stubborn agents, and rebels only to other rebels or to stubborn agents. From Theorem 1 in \cite{cao2024discrete}, any CS game admits at least one PNE. Thus, it suffices to restrict our analysis to the RS game. We define the potential function $\phi:A^N\to\mathbb{R}$ as follows:
$$
\phi(\boldsymbol{x}):=\frac{1}{2}(\sum_{k \in R}\left|\left\{j \in N_k: x_j \neq x_k\right\}\right|+\sum_{k \in R}\left|\left\{j \in N_k \cap S: x_j \neq x_k\right\}\right|).
$$
Since stubborn agents do not change their strategies, we only need to focus on the changes of strategies and utilities of rebels. For $\forall i \in R$, $\forall \boldsymbol{x}_{-i} \in A^N_{-i}$, $\forall x_i, x_i^{\prime} \in A$,
$$
\begin{aligned}
\phi\left(x_i, \boldsymbol{x}_{-i}\right)-\phi\left(x_i^{\prime}, \boldsymbol{x}_{-i}\right) =\left|\left\{j \in N_i: x_j \neq x_i\right\}\right|-\left|\left\{j \in N_i: x_j \neq x_i^{\prime}\right\}\right| =u_i\left(x_i, \boldsymbol{x}_{-i}\right)-u_i\left(x_i', \boldsymbol{x}_{-i}\right).
\end{aligned}
$$
Thus, we have $\phi(x_i,\boldsymbol{x}_{-i})-\phi(x_i^{\prime},\boldsymbol{x}_{-i})=u_i(x_i,\boldsymbol{x}_{-i})-u_i(x_i^{\prime},\boldsymbol{x}_{-i})$, which means the RS game is an exact potential game and admits at least one PNE.

For condition (ii), suppose there exists $a\in\{0,1\}$ such that each conformist has at least half of her neighbors in $C\cup S_a$. In this case, assigning $a$ to all conformists guarantees that no conformist has the incentive to deviate. Once the conformists’ strategies are fixed in this way, the game is reduced to an RS game and admits a PNE by the proof of condition (i).

For condition (iii), we let the rebels in $R_0$ adopt $0$ and those in $R_1$ adopt $1$. In this case, no rebel has the incentive to deviate. Once the rebels’ strategies are fixed in this way, the game is reduced to an CS game and admits a PNE by the proof of condition (i).
\end{proof}

\subsection{Proof of Corollary 1}
\begin{corollary}\label{co1}
A network HP game admits a PNE if one of the following conditions holds: (i) the game does not have coordinator–anti-coordinator (CoACo) edges; (ii) each coordinator $i\in Co$ has at least $\tau_i$ proportion of her neighbors from the set $Co \cup S_0$, or each coordinator has at least $1-\tau_i$ proportion of her neighbors from the set $Co \cup S_1$; (iii) there exists a partition of the anti-coordinator set $ACo$, $\{ACo_0, ACo_1\}$, such that each anti-coordinator $i\in ACo_0$ has at least $1-\tau_i$ proportion of her neighbors from the set $ACo_1 \cup S_1$ and each anti-coordinator $j\in ACo_1$ has at least $\tau_j$ proportion of her neighbors from the set $ACo_0 \cup S_0$.
\end{corollary}

\begin{proof}

We verify the existence of a PNE under each condition. 

For condition (i), any HP game can be transformed into its equivalent CRS game by adding stubborn neighbors to coordinators and anti-coordinators from Theorem \ref{th1}. Since the game does not have CoACo edges, its equivalent CRS game also does not have CR edges. Thus, the game admits a PNE from Theorem \ref{th2}. 

For condition (ii), each coordinator $i \in Co$ has at least $\tau_i$ proportion of neighbors in $Co \cup S_0$, or at least $1-\tau_i$ proportion of neighbors in $Co \cup S_1$. In this case, assigning strategy $0$ (resp. strategy $1$) to all coordinators ensures that no conformist has the incentive to deviate. Once the coordinators’ strategies are fixed in this way, the game is reduced to one with only anti-coordinators and stubborn agents, and admits a PNE by the proof of condition (i).

For condition (iii), we let the anti-coordinators in $ACo_0$ adopt $0$ and those in $ACo_1$ adopt $1$. In this case, no anti-coordinator has the incentive to deviate.  Once the anti-coordinators’ strategies are fixed in this way, the game is reduced to one with only coordinators and stubborn agents, and admits a PNE by the proof of condition (i).
\end{proof}

\subsection{Proof of Theorem 3}
\begin{theorem}\label{th3}
Consider a network CRS game on a random sparse graph with $n$ nodes and average degree $d$.
If conformists, rebels, stubborn agents with dominant strategy $0$, and stubborn agents with dominant strategy $1$ are randomly distributed on the network, then the game almost surely does not admit a PNE as $n \to \infty$.
\end{theorem}

\begin{proof}
We first identify a simple but critical local structure whose existence precludes the game from admitting a PNE. This local structure consists of 4 players, a degree-3 rebel with a degree-1 conformist neighbor and two stubborn neighbors whose dominant strategies are 0 and 1, respectively (see Figure \ref{FIGS1}). 
In this local structure, the rebel and the conformist will change their choices all the time: the rebel adopts the strategy contrary to the conformist, but then the conformist changes, making the rebel change, and so on. 

\begin{figure}[H]
\centering
\includegraphics[width=0.75\textwidth]{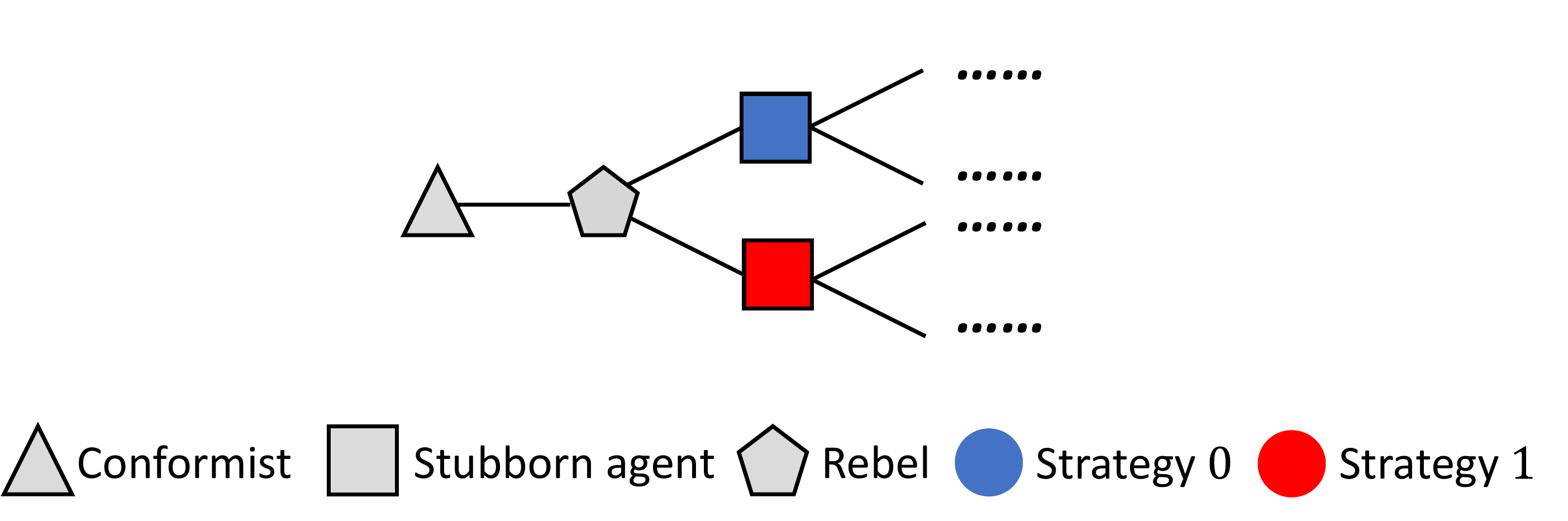}
\caption{The local structure proposed in the proof of Theorem 3 consists of 4 players, a degree-3 rebel with a degree-1 conformist neighbor and two stubborn neighbors whose dominant strategies are 0 and 1, respectively. Conformists, rebels, and stubborn agents are represented by triangles, pentagons, and squares, respectively.}\label{FIGS1}
\end{figure}

We next show that as the size of a CRS game increases, the probability that there is no such local structure tends to zero in a sparse random graph. Let $\alpha_C > 0$, $\alpha_R > 0$, $\alpha_{S_0} > 0$, and $\alpha_{S_1} > 0$ denote the proportions of conformists, rebels, stubborn agents with strategy 0, and stubborn agents with strategy 1, respectively. In a sparse random graph with $n$ nodes and average degree $d$, the degree distribution follows the Poisson distribution as $n\to\infty$. Thus, in a large CRS game, the probability that a player has degree 3 is  $\lim \limits_{n\to \infty}\binom{n-1}{3} (\frac{d}{n})^3 (1-\frac{d}{n})^{n-4}=\frac{d^3 e^{-d}}{6}$ and the expected number of degree-3 rebels is about $\mathcal{N}_3=\frac{d^3 e^{-d}}{6}\alpha_R n$. Furthermore, the probability that a degree-3 rebel forms the above local structure (i.e., it has a conformist neighbor with degree-1, and two stubborn neighbors with dominant strategies 0 and 1, respectively) is about $\kappa= 3! \alpha_{S_0} \alpha_{S_1} (d e^{-d}\alpha_C )$ for large $n$, where $d e^{-d}\alpha_C$ is the probability that one of the three neighbors is a degree-1 conformist.
Thus, the overall probability that none of these degree-3 rebels form the above local structure is $(1-\kappa)^{\mathcal{N}_3} \to 0$ as $n \to \infty.$ This implies that a CRS game on a sparse random graph almost surely does not admit a PNE as $n \to \infty$.
\end{proof}

We further validate our theoretical conclusions through agent-based simulations. Specifically, we generate sparse random networks with node counts ranging from 100 to 1000 at an interval of 100, where the average degrees $d$ are set to 4, 8, and 16, respectively. In each network, the four types of players are evenly distributed, each accounting for $1/4$ of the total players. For each sparse random network, we generate 100 distinct CRS network configurations by randomly assigning player types. The probability of admitting a PNE is then approximated by computing the proportion of these 100 configurations that converge within $n\times10^3$ time steps from three random initializations. Our simulation results show that the probability of PNE existence tends to 0 for $n \geq 800$, which is consistent with Theorem \ref{th3} (see Figure \ref{FIGS2}).

\begin{figure}[H]
\centering
\includegraphics[width=0.7\textwidth]{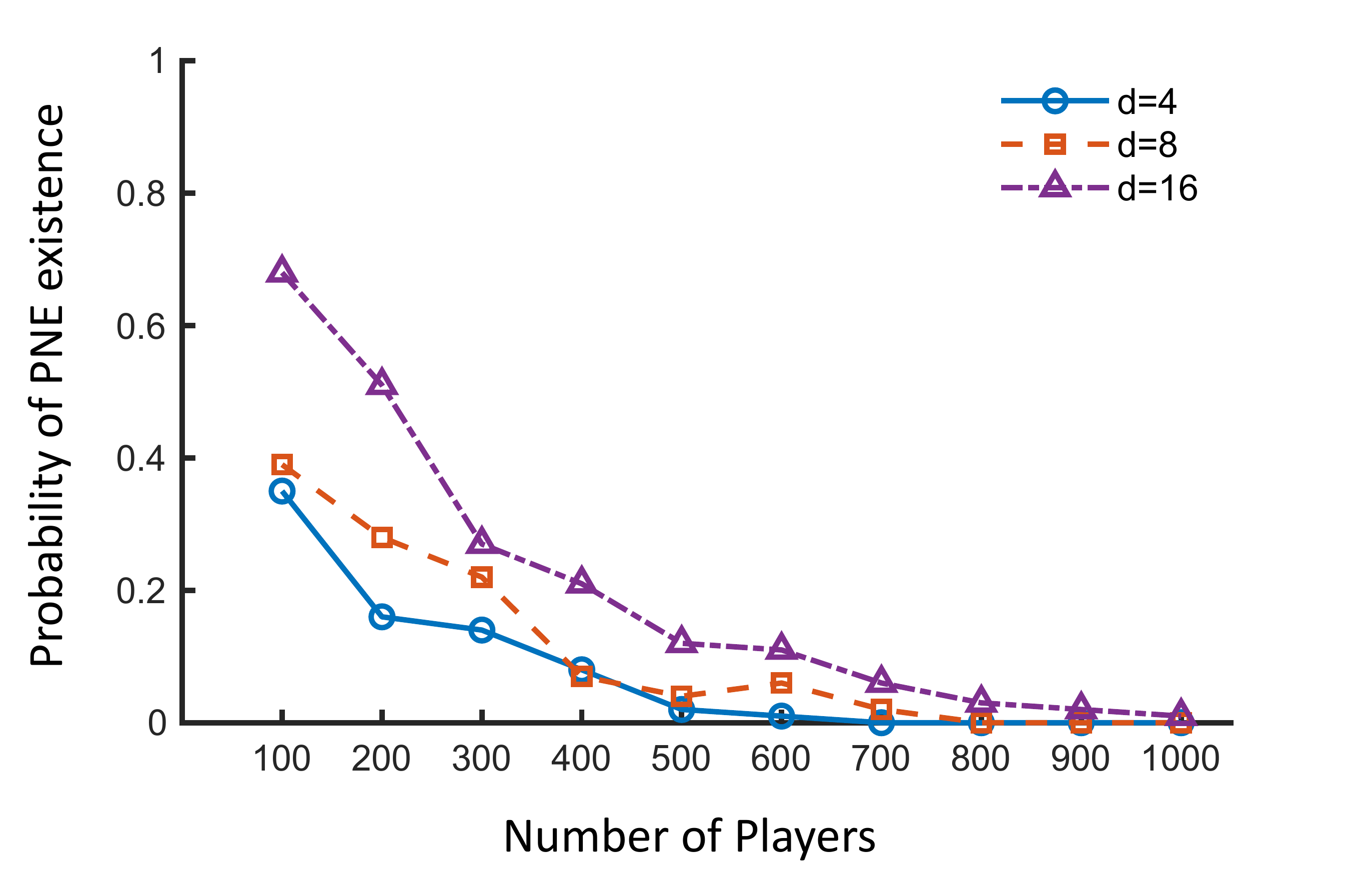}
\caption{Probability of PNE existence as the number of players varies across different average degrees. The blue solid line with circles, orange dashed line with squares, and purple dash-dotted line with triangles correspond to $d=4$, $d=8$, and $d=16$, respectively.}\label{FIGS2}
\end{figure}

\subsection{Derivation of Eq.(\ref{eq1})}
\begin{equation}
\begin{cases}
v_{C0} \approx h_{CC}x_c+h_{CR}x_r+h_{CS_{0}} \\
v_{C1} \approx h_{CC}(1-x_c) + h_{CR}(1-x_r) + h_{CS_{1}} \\
v_{R0} \approx h_{RR}(1-x_r) + h_{RC}(1-x_c) + h_{RS_{1}} \\
v_{R1} \approx h_{RR} x_r + h_{RC} x_c + h_{RS_{0}}
\end{cases}.
\tag{1}
\label{eq1}
\end{equation}

For a CRS game $G=(C,R,S,E,A,V)$, denote the frequencies of conformists, rebels, and stubborn agents adopting strategy 0 in the neighbors of player $i$ by $x_{i,c} = |C_0 \cap N_i|/|C \cap N_i|$, $x_{i,r} = |R_0 \cap N_i|/|R \cap N_i|$, and $x_{i,s} = |S_0 \cap N_i|/|S \cap N_i|$, respectively. Following these notations, the normalized utility of conformist $i\in C_0$ can be expressed as $h_{iC}x_{i,c} + h_{iR}x_{i,r} + h_{iS}x_{i,s}$, where $h_{iY}$ is the proportion of $Y \in \{C,R,S\}$ among $i$’s neighbors. The utilities for other types of players can be derived analogously.

In the derivation of \eqref{eq1}, we apply two types of mean-field approximations \citep[see e.g.][]{zhang2018fashion,PEI2024dynamic}. First, for each $i \in C$, we approximate her individual index $h_{iY}$ with the global index $h_{CY}$; similarly, for each $i \in R$, we approximate her individual index $h_{iY}$ with the global index $h_{RY}$. Second, for each player $i$, we approximate the local strategy distributions among her $C$, $R$, and $S$ neighbors, $x_{i,c}$, $x_{i,r}$, and $x_{i,s}$, with the corresponding global strategy distributions, $x_c$, $x_r$, and $x_s$, respectively. Therefore, $v_{C_0} = \frac{\sum_{i \in C_0} (h_{iC}x_{i,c} + h_{iR}x_{i,r} + h_{iS}x_{i,s})}{|C_0|} \approx h_{CC} \frac{\sum_{i \in C_0} x_{i,c}}{|C_0|} + h_{CR} \frac{\sum_{i \in C_0} x_{i,r}}{|C_0|} + h_{CS} \frac{\sum_{i \in C_0} x_{i,s}}{|C_0|} \approx h_{CC}x_c + h_{CR}x_r + h_{CS_0}$, giving the first formula in Equation (\ref{eq1}). The other three formulas can be similarly derived.

\subsection{Derivation of ODE.(\ref{eq2})}

We apply the stochastic approximation method \citep[see e.g.][]{sandholm2010population,roth2013stochastic} to derive ODE (\ref{eq2}). The asynchronous best response dynamics is a finite-state Markov process with transition probabilities
\begin{equation}
\begin{aligned}
P(n_{C0} \to n_{C0} - 1) &= f_C x_c \phi(v_{C0} - v_{C1}), \\
P(n_{C0} \to n_{C0} + 1) &= f_C (1 - x_c) \phi(v_{C1} - v_{C0}), \\
P(n_{R0} \to n_{R0} - 1) &= f_R x_r \phi(v_{R0} - v_{R1}), \\
P(n_{R0} \to n_{R0} + 1) &= f_R (1 - x_r) \phi(v_{R1} - v_{R0}).
\end{aligned}
\label{eqq4}
\end{equation}
In the first equation of (\ref{eqq4}), $P(n_{C0} \rightarrow n_{C0} - 1)$ denotes the probability that the population leaves state $(x_c, x_r)$ and enters $(x_c - \frac{1}{nf_C}, x_r)$ in one time step, where $f_C = \frac{n_C}{n}$ and $f_R = \frac{n_R}{n}$. This occurs if and only if a conformist who takes strategy $0$ is chosen to update (the probability is $f_C x_c$), and this player changes his/her action from $0$ to $1$ (the probability is $\phi(v_{C0} - v_{C1})$). The derivations of the other three equations are similar. Then the expected change of $x_c$ in each round can be approximated as $dx_c = \frac{P(n_{C0} \to n_{C0} + 1) - P(n_{C0} \to n_{C0} - 1)}{n f_C}.$ Similarly, the expected change of $x_r$ in each round can be approximated as $dx_r = \frac{P(n_{R0} \rightarrow n_{R0} + 1) - P(n_{R0} \rightarrow n_{R0} - 1)}{n f_R}.$ After a rescaling of time (i.e., let $dt = \frac{1}{n}$), we obtain ODE (\ref{eq2}).

\subsection{Derivation of Table 1}
\begin{equation}
\begin{cases}
\frac{dx_c}{dt} = (1 - x_c)\phi(v_{C1} - v_{C0}) - x_c\phi(v_{C0} - v_{C1}) \\
\frac{dx_r}{dt} = (1 - x_r)\phi(v_{R1} - v_{R0}) - x_r\phi(v_{R0} - v_{R1}) 
\end{cases},
\tag{2}
\label{eq2}
\end{equation}

Before proceeding to the fixed points and their stabilities of ODE (\ref{eq2}), we first provide a simple observation.

\begin{observation}\label{OB1}
The following statements for $\phi$ are true: \\
(a) at least one of $\phi(v_{C0} - v_{C1})$ and $\phi(v_{C1} - v_{C0})$ is zero. \\
(b) $\phi(v_{C0} - v_{C1}) = \phi(v_{C1} - v_{C0}) = 0$ if and only if $v_{C0} - v_{C1} = 0$, i.e., $h_{CC} x_c + h_{CR} x_r = x_c^*$, where $x_c^* = \frac{1}{2}-h_{CS_0}$.\\
(c) at least one of $\phi(v_{R0} - v_{R1})$ and $\phi(v_{R1} - v_{R0})$ is zero.\\
(d) $\phi(v_{R0} - v_{R1}) = \phi(v_{R1} - v_{R0}) = 0$ if and only if $v_{R0} - v_{R1} = 0$, i.e., $h_{RR} x_r + h_{RC} x_c = x_r^*$, where $x_r^* = \frac{1}{2}-h_{RS_0}$.
\end{observation}

Proof of the above observation is trivial, but it is useful for calculating the fixed points of ODE (\ref{eq2}). Observations 1(b) and 1(d) imply that the fixed points of ODE (\ref{eq2}) can be represented by eight parameters, $h_{CC}$, $h_{CR}$, $h_{CS_0}$, $h_{CS_1}$, $h_{RR}$, $h_{RC}$, $h_{RS_0}$, $h_{RS_1}$. In addition, an interior fixed point of ODE (\ref{eq2}) satisfies $v_{C0} - v_{C1} = v_{R0} - v_{R1} = 0$. 

\subsubsection*{Fixed points of ODE (\ref{eq2})}
We know that $(x_c, x_r) \in [0, 1]^2$ is a fixed point of ODE (\ref{eq2}) if and only if
\begin{equation}
x_c\phi(v_{C0} - v_{C1}) = (1 - x_c)\phi(v_{C1} - v_{C0}) = x_r\phi(v_{R0} - v_{R1}) = (1 - x_r)\phi(v_{R1} - v_{R0}) = 0.
\label{eqB1}
\end{equation}

We first calculate the interior fixed point $(\bar{x}_c, \bar{x}_r) \in (0, 1)^2$, which, due to \eqref{eqB1}, must satisfy Observations 1(b) and 1(d). If $h_{CC}h_{RR} \neq h_{CR}h_{RC}$, there is an interior fixed point
$$
\bar{x}_c = \frac{x_c^* h_{RR} - x_r^* h_{CR}}{h_{CC}h_{RR} - h_{CR}h_{RC}}, \quad \bar{x}_r = \frac{x_r^* h_{CC} - x_c^* h_{RC}}{h_{CC}h_{RR} - h_{CR}h_{RC}}.
$$
Furthermore, $(\bar{x}_c, \bar{x}_r)$ must satisfy $0 < \bar{x}_c < 1$ and $0 < \bar{x}_r < 1$. Thus, the existence conditions of an interior fixed point of ODE (\ref{eq2}) are as follows:

\begin{itemize}
\item If $h_{CC}h_{RR} < h_{CR}h_{RC}$, then $(\bar{x}_c, \bar{x}_r)$ is an interior fixed point if $x_c^* h_{RR} - x_r^* h_{CR} \in (h_{CC}h_{RR} - h_{CR}h_{RC}, 0)$ and $x_r^* h_{CC} - x_c^* h_{RC} \in (h_{CC}h_{RR} - h_{CR}h_{RC}, 0)$.

\item If $h_{CC}h_{RR} > h_{CR}h_{RC}$, then $(\bar{x}_c, \bar{x}_r)$ is an interior fixed point if $x_c^* h_{RR} - x_r^* h_{CR} \in (0, h_{CC}h_{RR} - h_{CR}h_{RC})$ and $x_r^* h_{CC} - x_c^* h_{RC} \in (0, h_{CC}h_{RR} - h_{CR}h_{RC})$.

\item If $h_{CC}h_{RR} = h_{CR}h_{RC}$, when $\frac{x_c^*}{h_{CC}} = \frac{x_r^*}{h_{RC}}$, ODE (\ref{eq2}) has a continuum of fixed points $x_r = \frac{x_c^*}{h_{CR}} - \frac{h_{CC}}{h_{CR}} x_c$.
\end{itemize}

We now calculate the boundary fixed points. At the boundary $x_c = 0$, due to \eqref{eqB1}, $(0, x_r)$ is a fixed point if and only if
\begin{equation}
\begin{cases}
\phi(v_{C1} - v_{C0}) = 0, \\
x_r\phi(v_{R0} - v_{R1}) = (1 - x_r)\phi(v_{R1} - v_{R0}) = 0.
\end{cases}
\label{eqB3}
\end{equation}
The first equation in (\ref{eqB3}) is equivalent to $v_{C1} - v_{C0} \geq 0$, i.e., $2h_{CR}x_r \leq h_{CC} + h_{CR} + h_{CS_1} - h_{CS_0}$. The solution of the second equation in (\ref{eqB3}) should satisfy one of the following three conditions:
\begin{equation}
\begin{cases}
(i) & x_r = 1 \quad \text{and} \quad v_{R0} - v_{R1} \geq 0, \\
(ii) & x_r = 0 \quad \text{and} \quad v_{R1} - v_{R0} \geq 0, \\
(iii) & v_{R0} - v_{R1} = 0,
\end{cases}
\label{eqB4}
\end{equation}
where condition (i) implies $h_{RC} + h_{RS_1} -h_{RR} - h_{RS_0} \geq 2h_{RC}x_c$, condition (ii) implies $2h_{RC}x_c \geq h_{RC} + h_{RS_1} + h_{RR} - h_{RS_0}$, and condition (iii) implies $2h_{RR}x_r + 2h_{RC}x_c = h_{RR} + h_{RC} + h_{RS_1} - h_{RS_0}$. Thus, we have:

\begin{itemize}
\item $(0, 1)$ is a boundary fixed point if and only if $h_{CR} \leq h_{CC} - h_{CS_0} + h_{CS_1}$ and $h_{RC} \geq h_{RR} + h_{RS_0} - h_{RS_1}$;

\item $(0, 0)$ is a boundary fixed point if and only if $h_{CR} \geq -h_{CC} + h_{CS_0} - h_{CS_1}$ and $h_{RC} \leq -h_{RR} + h_{RS_0} - h_{RS_1}$;

\item $(0, \frac{x_r^*}{h_{RR}}) = (0, \frac{h_{RR} + h_{RC} + h_{RS_1} - h_{RS_0}}{2h_{RR}})$ is a boundary fixed point if and only if $-h_{RR} + h_{RS_0} - h_{RS_1} < h_{RC} < h_{RR} + h_{RS_0} - h_{RS_1}$ and $\frac{1-2h_{RS_0}}{h_{RR}} \leq \frac{1-2h_{CS_0}}{h_{CR}}$.
\end{itemize}

At the boundary $x_c = 1$, due to \eqref{eqB1}, $(1, x_r)$ is a fixed point if and only if
\begin{equation}
\begin{cases} 
\phi(v_{C0} - v_{C1}) = 0, \\ 
x_r\phi(v_{R0} - v_{R1}) = (1 - x_r)\phi(v_{R1} - v_{R0}) = 0.
\end{cases}
\label{eqB5}
\end{equation}
The first equation in (\ref{eqB5}) is equivalent to $v_{C0} - v_{C1} \geq 0$, i.e., $2h_{CR}x_r \geq - h_{CC} + h_{CR} + h_{CS_1} - h_{CS_0}$. Similarly, the solution of the second equation in (\ref{eqB5}) should satisfy one of the three conditions of (\ref{eqB4}). Thus, we have:

\begin{itemize}
\item $(1, 1)$ is a boundary fixed point if and only if $h_{CR} \geq -h_{CC} - h_{CS_0} + h_{CS_1}$ and $h_{RC} \leq -h_{RR} - h_{RS_0} + h_{RS_1}$;

\item $(1, 0)$ is a boundary fixed point if and only if $h_{CR} \leq h_{CC} + h_{CS_0} - h_{CS_1}$ and $h_{RC} \geq h_{RR} - h_{RS_0} + h_{RS_1}$;

\item $(1, \frac{x_r^*-h_{RC}}{h_{RR}}) = (1, \frac{h_{RR} - h_{RC} + h_{RS_1} - h_{RS_0}}{2h_{RR}})$ is a boundary fixed point if and only if $-h_{RR} - h_{RS_0} + h_{RS_1} < h_{RC} < h_{RR} - h_{RS_0} + h_{RS_1}$ and $\frac{1-2h_{RS_1}}{h_{RR}} \leq \frac{1-2h_{CS_1}}{h_{CR}}$.
\end{itemize}

At the boundary $x_r = 0$, due to \eqref{eqB1}, $(x_c, 0)$ is a fixed point if and only if
\begin{equation}
\begin{cases}
\phi(v_{R1} - v_{R0}) = 0, \\
x_c\phi(v_{C0} - v_{C1}) = (1 - x_c)\phi(v_{C1} - v_{C0}) = 0.
\end{cases}
\label{eqB6}
\end{equation}
The first equation in (\ref{eqB6}) is equivalent to $v_{R1} - v_{R0} \geq 0$, i.e., $2h_{RC}x_c \geq h_{RC} + h_{RR} + h_{RS_1} - h_{RS_0}$. The solution of the second equation in (\ref{eqB6}) should satisfy one of the following three conditions:
\begin{equation}
\begin{cases}
(i) & x_c = 1 \quad \text{and} \quad v_{C0} - v_{C1} \geq 0, \\
(ii) & x_c = 0 \quad \text{and} \quad v_{C1} - v_{C0} \geq 0, \\
(iii) & v_{C0} - v_{C1} = 0,
\end{cases}
\label{eqB7}
\end{equation}
where condition (i) implies $2h_{CR}x_r \geq -h_{CC} + h_{CR} + h_{CS_1} - h_{CS_0}$, condition (ii) implies $2h_{CR}x_r \leq h_{CC} + h_{CR} + h_{CS_1} - h_{CS_0}$, and condition (iii) implies $2h_{CR}x_r + 2h_{CC}x_c = h_{CC} + h_{CR} + h_{CS_1} - h_{CS_0}$. Thus, we have:
\begin{itemize}
\item $(1, 0)$ is a boundary fixed point if and only if $h_{CR} \leq h_{CC} + h_{CS_0} - h_{CS_1}$ and $h_{RC} \geq h_{RR} - h_{RS_0} + h_{RS_1}$;
\item $(0, 0)$ is a boundary fixed point if and only if $h_{CR} \geq -h_{CC} + h_{CS_0} - h_{CS_1}$ and $h_{RC} \leq -h_{RR} + h_{RS_0} - h_{RS_1}$;
\item $(\frac{x_c^*}{h_{CC}}, 0) = (\frac{h_{CC} + h_{CR} + h_{CS_1} - h_{CS_0}}{2h_{CC}}, 0)$ is a boundary fixed point if and only if $-h_{CC} + h_{CS_0} - h_{CS_1} < h_{CR} < h_{CC} + h_{CS_0} - h_{CS_1}$ and $\frac{1-2h_{RS_0}}{h_{RC}} \leq \frac{1-2h_{CS_0}}{h_{CC}}$.
\end{itemize}

At the boundary $x_r = 1$, due to \eqref{eqB1}, $(x_c, 1)$ is a fixed point if and only if 
\begin{equation}
\begin{cases} 
\phi(v_{R0} - v_{R1}) = 0, \\ 
x_c\phi(v_{C0} - v_{C1}) = (1 - x_c)\phi(v_{C1} - v_{C0}) = 0.
\label{eqB8}
\end{cases}
\end{equation}
The first equation in \eqref{eqB8} is equivalent to $v_{R0} - v_{R1} \geq 0$, i.e., $2h_{RC}x_c \leq h_{RC} - h_{RR} + h_{RS_1} - h_{RS_0}$. Similarly, the solution of the second equation in \eqref{eqB8} should satisfy one of the three conditions of \eqref{eqB7}. Thus, we have:  
\begin{itemize}
\item $(1, 1)$ is a boundary fixed point if and only if $h_{CR} \geq -h_{CC} - h_{CS_0} + h_{CS_1}$ and $h_{RC} \leq -h_{RR} - h_{RS_0} + h_{RS_1}$;  
\item $(0, 1)$ is a boundary fixed point if and only if $h_{CR} \leq h_{CC} - h_{CS_0} + h_{CS_1}$ and $h_{RC} \geq h_{RR} + h_{RS_0} - h_{RS_1}$;  
\item $(\frac{x_c^*-h_{CR}}{h_{CC}}, 1) = (\frac{h_{CC} - h_{CR} + h_{CS_1} - h_{CS_0}}{2h_{CC}}, 1)$ is a boundary fixed point if and only if $-h_{CC} - h_{CS_0} + h_{CS_1} < h_{CR} < h_{CC} - h_{CS_0} + h_{CS_1}$ and $\frac{1-2h_{RS_1}}{h_{RC}} \leq \frac{1-2h_{CS_1}}{h_{CC}}$.
\end{itemize}

\subsubsection*{Local stability analysis}
We analyze the local stabilities of all the interior and boundary fixed points. Let $\textbf{J}(x_c, x_r)$ be the Jacobian matrix of ODE (\ref{eq2}) at $(x_c, x_r)$. The entries of $\textbf{J}$ can be calculated as

$$
\begin{cases}
\begin{aligned}
\textbf{J}_{11}(x_c, x_r) &= -\phi(v_{C0} - v_{C1}) - \phi(v_{C1} - v_{C0}) - 2h_{CC} x_c \phi'(v_{C0} - v_{C1}) - 2h_{CC}(1 - x_c)\phi'(v_{C1} - v_{C0}), \\
\textbf{J}_{12}(x_c, x_r) &= -2 h_{CR} x_c \phi'(v_{C0} - v_{C1}) - 2 h_{CR} (1 - x_c)\phi'(v_{C1} - v_{C0}), \\
\textbf{J}_{21}(x_c, x_r) &= 2 h_{RC} x_r \phi'(v_{R0} - v_{R1}) + 2 h_{RC} (1 - x_r)\phi'(v_{R1} - v_{R0}), \\
\textbf{J}_{22}(x_c, x_r) &= -\phi(v_{R0} - v_{R1}) - \phi(v_{R1} - v_{R0}) + 2h_{RR} x_r \phi'(v_{R0} - v_{R1}) + 2h_{RR}(1 - x_r)\phi'(v_{R1} - v_{R0}). 
\end{aligned}
\end{cases}
$$

$\bullet(\bar{x}_c, \bar{x}_r)$: At the interior fixed point \((\bar{x}_c, \bar{x}_r)\), the Jacobian matrix of ODE (\ref{eq2}) is given by
\[
\begin{pmatrix}
-2h_{CC} \phi'(0) & -2 h_{CR}\phi'(0) \\
2 h_{RC}\phi'(0) & 2h_{RR} \phi'(0)
\end{pmatrix},
\]
with eigenvalues 

$\begin{aligned}
\lambda_{1,2} = & - h_{CC}\phi'(0) + h_{RR}\phi'(0) \pm  \sqrt{(h_{CC}\phi'(0) - h_{RR}\phi'(0))^2 - 4(h_{RC}h_{CR} - h_{CC}h_{RR}) \phi'(0)\phi'(0)}.
\end{aligned}$ The interior fixed point $(\bar{x}_c, \bar{x}_r)$ is locally asymptotically stable if and only if both eigenvalues have negative real parts. Recall that $\phi'(0)$ is the left-hand derivative of $\phi$ at 0, which is assumed to be negative. Hence $(\bar{x}, \bar{y})$ is locally asymptotically stable if and only if $ h_{CC} < h_{RR}$ and $h_{CC}h_{RR} < h_{CR}h_{RC}$.

$\bullet(0,0)$: When $h_{CR} \geq -h_{CC} + h_{CS_0} - h_{CS_1}$ and $h_{RC} \leq -h_{RR} + h_{RS_0} - h_{RS_1}$, ODE (\ref{eq2}) has a boundary fixed point $(0,0)$. We have $v_{C0} - v_{C1} = h_{CS_0} - h_{CC} - h_{CS_1} - h_{CR} \leq 0$, so $\phi(v_{C0} - v_{C1}) \geq 0$, $\phi(v_{C1} - v_{C0}) = 0$. Similarly, $v_{R0} - v_{R1} = h_{RR} + h_{RS_1} + h_{RC} - h_{RS_0} \leq 0$, $\phi(v_{R0} - v_{R1}) \geq 0$, $\phi(v_{R1} - v_{R0}) = 0$. The Jacobian matrix at this fixed point is
\begin{equation}
\begin{pmatrix}
-\phi(v_{C0} - v_{C1}) & 0 \\
0 & -\phi(v_{R0} - v_{R1})
\label{eqB9}
\end{pmatrix}.
\end{equation}
The two eigenvalues of \eqref{eqB9} are $\lambda_1 = -\phi(v_{C0} - v_{C1}) \leq 0$ and $\lambda_2 = -\phi(v_{R0} - v_{R1}) \leq 0$. Hence $(0,0)$ is locally asymptotically stable if and only if $h_{CR} > -h_{CC} + h_{CS_0} - h_{CS_1}$ and $h_{RC} < -h_{RR} + h_{RS_0} - h_{RS_1}$.

$\bullet(1,1)$: When $h_{CR} \geq -h_{CC} - h_{CS_0} + h_{CS_1}$ and $h_{RC} \leq -h_{RR} - h_{RS_0} + h_{RS_1}$, ODE (\ref{eq2}) has a boundary fixed point $(1,1)$. We have $v_{C0} - v_{C1} = h_{CC} + h_{CS_0} + h_{CR} - h_{CS_1} \geq 0$, so $\phi(v_{C0} - v_{C1}) = 0$, $\phi(v_{C1} - v_{C0}) \geq 0$. Similarly, $v_{R0} - v_{R1} = h_{RS_1} - h_{RR} - h_{RS_0} - h_{RC} \geq 0$, so $\phi(v_{R0} - v_{R1}) = 0$, $\phi(v_{R1} - v_{R0}) \geq 0$. The Jacobian matrix at this fixed point is
\begin{equation}
\begin{pmatrix}
-\phi(v_{C1} - v_{C0}) & 0 \\
0 & -\phi(v_{R1} - v_{R0})
\label{eqB10}
\end{pmatrix}.
\end{equation}
The two eigenvalues of \eqref{eqB10} are $\lambda_1 = -\phi(v_{C1} - v_{C0}) \leq 0$ and $\lambda_2 = -\phi(v_{R1} - v_{R0}) \leq 0$. Hence $(1,1)$ is locally asymptotically stable if and only if $h_{CR} > -h_{CC} - h_{CS_0} + h_{CS_1}$ and $h_{RC} < -h_{RR} - h_{RS_0} + h_{RS_1}$.

$\bullet(0,1)$: When $h_{CR} \leq h_{CC} - h_{CS_0} + h_{CS_1}$ and $h_{RC} \geq h_{RR} + h_{RS_0} - h_{RS_1}$, ODE (\ref{eq2}) has a boundary fixed point $(0,1)$. We have $v_{C0} - v_{C1} = h_{CR} +h_{CS_0} - h_{CC} - h_{CS_1} \leq 0$, so $\phi(v_{C0} - v_{C1}) \geq 0$, $\phi(v_{C1} - v_{C0}) = 0$. Similarly, $v_{R0} - v_{R1} = h_{RC} + h_{RS_1} - h_{RR} - h_{RS_0} \geq 0 $, so $\phi(v_{R0} - v_{R1}) = 0$, $\phi(v_{R1} - v_{R0}) \geq 0$. The Jacobian matrix at this point is
\begin{equation}
\begin{pmatrix}
-\phi(v_{C0} - v_{C1}) & 0 \\
0 & -\phi(v_{R1} - v_{R0})
\label{eqB11}
\end{pmatrix}.
\end{equation}
The two eigenvalues of \eqref{eqB11} are $\lambda_1 = -\phi(v_{C0} - v_{C1}) \leq 0$ and $\lambda_2 = -\phi(v_{R1} - v_{R0}) \leq 0$. Hence $(0,1)$ is locally asymptotically stable if and only if $h_{CR} < h_{CC} - h_{CS_0} + h_{CS_1}$ and $h_{RC} > h_{RR} + h_{RS_0} - h_{RS_1}$.

$\bullet(1,0)$: When $h_{CR} \leq h_{CC} + h_{CS_0} - h_{CS_1}$ and $h_{RC} \geq h_{RR} - h_{RS_0} + h_{RS_1}$, ODE (\ref{eq2}) has a boundary fixed point $(1,0)$. We have $v_{C0} - v_{C1} = h_{CC} + h_{CS_0} - h_{CR} - h_{CS_1} \geq 0$, so $\phi(v_{C0} - v_{C1}) = 0$, $\phi(v_{C1} - v_{C0}) \geq 0$. Similarly, $v_{R0} - v_{R1} = h_{RS_1} + h_{RR} - h_{RS_0} - h_{RC} \leq 0$, so $\phi(v_{R0} - v_{R1}) \geq 0$, $\phi(v_{R1} - v_{R0}) = 0$. The Jacobian matrix at this fixed point is
\begin{equation}
\begin{pmatrix}
-\phi(v_{C1} - v_{C0}) & 0 \\
0 & -\phi(v_{R0} - v_{R1})
\label{eqB12}
\end{pmatrix}.
\end{equation}
The two eigenvalues of \eqref{eqB12} are $\lambda_1 = -\phi(v_{C1} - v_{C0}) \leq 0$ and $\lambda_2 = -\phi(v_{R0} - v_{R1}) \leq 0$. Hence $(1,0)$ is locally asymptotically stable if and only if $h_{CR} < h_{CC} + h_{CS_0} - h_{CS_1}$ and $h_{RC} > h_{RR} - h_{RS_0} + h_{RS_1}$.

$\bullet(0, \frac{x_r^*}{h_{RR}})$: When $-h_{RR} + h_{RS_0} - h_{RS_1} < h_{RC} < h_{RR} + h_{RS_0} - h_{RS_1}$ and $\frac{1-2h_{RS_0}}{h_{RR}} \leq \frac{1-2h_{CS_0}}{h_{CR}}$, ODE (\ref{eq2}) has a boundary fixed point $(0, \frac{x_r^*}{h_{RR}})$. We have $v_{C0} - v_{C1} = 2\frac{h_{CR}}{h_{RR}}x_r^*-2x_c^* \leq 0$, so $\phi(v_{C0} - v_{C1}) \geq 0$, $\phi(v_{C1} - v_{C0}) = 0$. Similarly, $v_{R0} - v_{R1} = 0 $, so $\phi(v_{R0} - v_{R1}) = \phi(v_{R1} - v_{R0}) = 0$. The Jacobian matrix at this point is
\begin{equation}
\begin{pmatrix}
-\phi(v_{C0} - v_{C1}) & 0 \\
2h_{RC}\phi'(0) & 2h_{RR}\phi'(0)
\label{eqB13}
\end{pmatrix}.
\end{equation}
The two eigenvalues of \eqref{eqB13} are $\lambda_1 = -\phi(v_{C0} - v_{C1}) \leq 0$ and $\lambda_2 = 2h_{RR}\phi'(0) < 0$. Hence $(0,\frac{x_r^*}{h_{RR}})$ is locally asymptotically stable if and only if $\frac{1-2h_{RS_0}}{h_{RR}} < \frac{1-2h_{CS_0}}{h_{CR}}$.

$\bullet(\frac{x_c^*}{h_{CC}},0)$: When $-h_{CC} + h_{CS_0} - h_{CS_1} < h_{CR} < h_{CC} + h_{CS_0} - h_{CS_1}$ and $\frac{1-2h_{RS_0}}{h_{RC}} \leq \frac{1-2h_{CS_0}}{h_{CC}}$, ODE (\ref{eq2}) has a boundary fixed point $(\frac{x_c^*}{h_{CC}},0)$. We have $v_{C0} - v_{C1} = 0$, so $\phi(v_{C0} - v_{C1}) = \phi(v_{C1} - v_{C0}) = 0$. Similarly, $v_{R0} - v_{R1} = 2x_r^* - 2\frac{h_{RC}}{h_{CC}}x_c^* \leq 0$, so $\phi(v_{R0} - v_{R1}) \geq 0$, $\phi(v_{R1} - v_{R0}) = 0$. The Jacobian matrix at this fixed point is
\begin{equation}
\begin{pmatrix}
-2h_{CC}\phi'(0) & -2h_{CR}\phi'(0) \\
0 & -\phi(v_{R0} - v_{R1})
\label{eqB14}
\end{pmatrix}.
\end{equation}
The two eigenvalues of \eqref{eqB14} are $\lambda_1 = -2h_{CC}\phi'(0) > 0$ and $\lambda_2 = -\phi(v_{R0} - v_{R1}) \leq 0$. Hence $(\frac{x^*}{h_{CC}},0)$ is unstable.

$\bullet(1, \frac{x_r^*-h_{RC}}{h_{RR}})$: When $-h_{RR} - h_{RS_0} + h_{RS_1} < h_{RC} < h_{RR} - h_{RS_0} + h_{RS_1}$ and $\frac{1-2h_{RS_1}}{h_{RR}} \leq \frac{1-2h_{CS_1}}{h_{CR}}$, ODE (\ref{eq2}) has a boundary fixed point $(1, \frac{x_r^*-h_{RC}}{h_{RR}})$. We have $v_{C0} - v_{C1} = h_{CR}(\frac{h_{RS_1} - h_{RC} - h_{RS_0}}{h_{RR}} - \frac{h_{CS_1} - h_{CC} - h_{CS_0}}{h_{CR}}) \geq 0$, so $\phi(v_{C0} - v_{C1}) = 0$, $\phi(v_{C1} - v_{C0}) \geq 0$. Similarly, $v_{R0} - v_{R1} = 0 $, so $\phi(v_{R0} - v_{R1}) = \phi(v_{R1} - v_{R0}) = 0$. The Jacobian matrix at this point is
\begin{equation}
\begin{pmatrix}
-\phi(v_{C1} - v_{C0}) & 0 \\
2h_{RC}\phi'(0) & 2h_{RR}\phi'(0)
\label{eqB15}
\end{pmatrix}.
\end{equation}
The two eigenvalues of \eqref{eqB15} are $\lambda_1 = -\phi(v_{C1} - v_{C0}) \leq 0$ and $\lambda_2 = 2h_{RR}\phi'(0) < 0$. Hence $(1, \frac{x_r^*-h_{RC}}{h_{RR}})$ is locally asymptotically stable if and only if $\frac{1-2h_{RS_1}}{h_{RR}} < \frac{1-2h_{CS_1}}{h_{CR}}$.

$\bullet(\frac{x_c^*-h_{CR}}{h_{CC}},1)$: When $-h_{CC} - h_{CS_0} + h_{CS_1} < h_{CR} < h_{CC} - h_{CS_0} + h_{CS_1}$ and $\frac{1-2h_{RS_1}}{h_{RC}} \leq \frac{1-2h_{CS_1}}{h_{CC}}$, ODE (\ref{eq2}) has a boundary fixed point $(\frac{x_c^*-h_{CR}}{h_{CC}},1)$. We have $v_{C0} - v_{C1} = 0$, so $\phi(v_{C0} - v_{C1}) = \phi(v_{C1} - v_{C0}) = 0$. Similarly, $v_{R0} - v_{R1} = h_{RC}(\frac{h_{RS_1} - h_{RR} - h_{RS_0}}{h_{RC}} - \frac{h_{CS_1} - h_{CR} - h_{CS_0}}{h_{CC}}) \geq 0$, so $\phi(v_{R0} - v_{R1}) = 0$, $\phi(v_{R1} - v_{R0}) \geq 0$. The Jacobian matrix at this fixed point is
\begin{equation}
\begin{pmatrix}
-2h_{CC}\phi'(0) & -2h_{CR}\phi'(0) \\
0 & -\phi(v_{R1} - v_{R0})
\label{eqB16}
\end{pmatrix}.
\end{equation}
The two eigenvalues of \eqref{eqB16} are $\lambda_1 = -2h_{CC}\phi'(0) > 0$ and $\lambda_2 = -\phi(v_{R1} - v_{R0}) \leq 0$. Hence $(\frac{x_c^*-h_{CR}}{h_{CC}},1)$ is unstable.

Finally, we summarize the existence and stability conditions of all fixed points of ODE (\ref{eq2}) in Table 1.

\subsection{Proof of Theorem 4}
\begin{theorem}\label{th4}
(i) If $h_{CR} h_{RC} < h_{CC} h_{RR}$, then ODE (\ref{eq2}) has stable boundary fixed points. (ii) If $h_{CR} h_{RC} > h_{CC} h_{RR}$, $h_{CC} < h_{RR}$, $0 < \bar{x}_c < 1$ and $0 < \bar{x}_r < 1$, then ODE (\ref{eq2}) has a stable interior fixed point.
(iii) If $h_{CR} h_{RC} > h_{CC} h_{RR}$, $h_{CC} > h_{RR}$, $0 < \bar{x}_c < 1$ and $0 < \bar{x}_r < 1$, then a limit cycle emerges.
\end{theorem}

\begin{proof}
(i) From $\frac{dx_c}{dt} = 0$ and $\frac{dx_r}{dt} = 0$, the zero-isoclines in the $(x_c, x_r)$-plane can be expressed as 
$$
L_c: x_r = -\frac{h_{CC}}{h_{CR}}x_c + \frac{x_c^*}{h_{CR}}, L_r: x_r = -\frac{h_{RC}}{h_{RR}}x_c + \frac{x_r^*}{h_{RR}}.
$$
Denote the slopes of $L_c$ and $L_r$ as $l_c$ and $l_r$, respectively. It is straightforward that both slopes are non-positive. If the intersection $(\bar{x}, \bar{y}) \notin [0,1]^2$, ODE (\ref{eq2}) clearly has no interior fixed points or limit cycles. If $(\bar{x}, \bar{y}) \in [0,1]^2$, the last row of Table 1 indicates that the interior fixed point $(\bar{x}, \bar{y})$ is unstable when $h_{CR}h_{RC} < h_{CC}h_{RR}$. Furthermore, vector field analysis of the dynamical system shows that trajectories rotate clockwise around the interior fixed point if and only if $l_c > l_r$, or equivalently $h_{CR}h_{RC} > h_{CC}h_{RR}$. Thus, ODE (\ref{eq2}) has neither stable interior fixed points nor limit cycles when $h_{CR}h_{RC} < h_{CC}h_{RR}$. Since ODE (\ref{eq2}) is smooth on $[0,1]^2$, the Poincaré-Bendixson theorem implies the existence of stable boundary fixed points.  

(ii) The conclusion follows directly from Table 1.  

(iii) When $0 < \bar{x}_c < 1$ and $0 < \bar{x}_r < 1$, if $h_{CC} > h_{RR}$ and $h_{CR}h_{RC} > h_{CC}h_{RR}$, the interior fixed point $(\bar{x}, \bar{y})$ is unstable and trajectories rotate clockwise around it, leading to the emergence of a limit cycle. 
\end{proof}

\subsection{Proof of Theorem 5}
\begin{theorem}\label{th5}
A network HP game admits an L-PNE if and only if the following three conditions hold: (i) the game does not have CoACo edges; (ii) for each subnetwork composed of coordinators, all stubborn neighbors adopt the same strategy; (iii) for each subnetwork composed of anti-coordinators ${ACo}^k$, there exists a partition, $\{ACo_0^k, ACo_1^k\}$, such that no two players within the same set are connected by an edge and all stubborn neighbors of players in set $ACo_{0}^k$ (or $ACo_{1}^k$) adopt strategy $1$ (or $0$).
\end{theorem}

\begin{proof}
Sufficiency. Since the game does not have CoACo edges, it suffices to prove that an L-PNE exists in the coordinator subnetwork and the anti-coordinator subnetwork, respectively. Note that at an L-PNE, each coordinator adopts the same strategy as her neighbors. From Theorem 2 in \cite{cao2024discrete}, condition (ii) implies an L-PNE is guaranteed to exist in each coordinator subnetwork. Moreover, condition (iii) implies that in each anti-coordinator subnetwork ${ACo}^k$, the strategy profile where all anti-coordinators in $ACo_0^k$ adopt strategy $0$ and all those in $ACo_1^k$ adopt strategy $1$ constitutes an L-PNE. In summary, any game satisfying the three conditions admits an L-PNE. \textcolor{blue}{In particular, as each subnetwork is endowed with stubborn neighbors, coordinators within each coordinator subnetwork necessarily adopt the same strategies as their stubborn neighbors, and the partition of every anti-coordinator subnetwork is unique, which together ensures the uniqueness of the L-PNE.}

Necessity. We prove by contradiction that an HP game does not have an L-PNE if any of the aforementioned conditions is not satisfied. For condition (i), if the network has a CoACo edge, then the game does not admit an L-PNE, as neighboring coordinators and anti-coordinators cannot be satisfied simultaneously. For condition (ii), if stubborn neighbors adopt different strategies in some subnetworks composed of coordinators, then the game does not admit L-PNE, as follows directly from Theorem 2 in \cite{cao2024discrete}. For condition (iii), if a subnetwork composed of anti-coordinators does not admit a bipartition where no two players within the same set are connected by an edge, then this subnetwork must contain an odd-length cycle. Since graphs containing odd-length cycles are not 2-colorable, the game admits no L-PNE. Alternatively, if such a bipartition with no intra-subset edges exists, yet the stubborn neighbors of anti-coordinators in one subset adopt both strategies $1$ and $0$, the game does not have L-PNE since anti-coordinators must adopt strategies distinct from those of their neighbors at any L-PNE.
\end{proof}

\subsection{Proof of Corollary 4}

\setcounter{corollary}{3}
\begin{corollary}\label{co4}
In the case of limited information, the asynchronous best response dynamics either converges to an L-PNE (if it exists) or to a stationary distribution.
\end{corollary}

\begin{proof}
We first prove if the game admits an L-PNE, the asynchronous best response dynamics converges to it. From Theorem \ref{th5}, each connected subnetwork resulting from the removal of stubborn agents consists entirely of coordinators or anti-coordinators if the game possesses an L-PNE. Thus, we only need to consider the global stability of the L-PNE in the coordinator subnetwork ${Co}^k$ and the anti-coordinator subnetwork ${ACo}^k$, respectively. 

Without loss of generality, we assume stubborn agents are connected to each ${Co}^k$ and ${ACo}^k$ (their existence affects the uniqueness of the L-PNE, and the proof proceeds similarly for the case where there are no stubborn agents) and denote their strategy as $a \in A$. We denote $0<p_j<1$ as the probability that player $j$ updates her strategy at each time step. For $\forall j \in {Co}^k$, $\exists\ i \in S$ with $x_i = a$ and $c_{j,1}, \dots, c_{j,l(j)} \in {Co}^k$ such that $q_{j c_{j,1}} q_{c_{j,1} c_{j,2}} \cdots q_{c_{j,l(j)} i} > 0$ (i.e., coordinator $j$ can observe stubborn agent $i$ through an information path of $l(j)$ coordinators). The probability that the asynchronous best response dynamics converges to the L-PNE (all coordinators adopt the same strategy as stubborn neighbors) in subnetwork ${Co}^k$ in finite steps from any initial state is at least $\prod_{j \in {Co}^k} ( ( p_{c_{j,l(j)}} q_{c_{j,l(j)} i} \prod_{v \neq i} ( 1 - q_{c_{j,l(j)} v} ) ) \cdots ( p_j q_{j c_{j,1}} \prod_{v \neq c_{j,1}} ( 1 - q_{c_{j,1} v} ) ) ) > 0$, where this probability describes a process in which $c_{j,l(j)}$ updates first and adopts the strategy of $i$, then $c_{j,l(j)-1}$ updates and adopts strategy $c_{j,l(j)}$, and so on until $j$ updates last and adopts strategy $c_{j,1}$, with this entire sequence repeated for each $j \in {Co}^k$. Thus, the unique L-PNE in ${Co}^k$ is globally convergent. Similarly, for $\forall j \in {ACo}_a^k$ (or ${ACo}_{1-a}^k$), $\exists\ i \in S$ with $x_i = a$ and $r_{j,1}, \dots, r_{j,l(j)} \in {ACo}^k$, such that $q_{j r_{j,1}} q_{r_{j,1} r_{j,2}} \cdots q_{r_{j,l(j)} i} > 0$. The probability that the asynchronous best response dynamics converges to the unique L-PNE in subnetwork ${ACo}^k$ in finite steps from any initial state is at least $\prod_{j \in {ACo}^k} ( ( p_{r_{j,l(j)}} q_{r_{j,l(j)} i} \prod_{v \neq i} ( 1 - q_{r_{j,l(j)} v} ) ) \cdots ( p_j q_{j r_{j,1}} \prod_{v \neq r_{j,1}} ( 1 - q_{r_{j,1} v} ) ) ) > 0$, where this probability admits an analogous interpretation to that for the subnetwork ${Co}^k$. Thus, the unique L-PNE in ${ACo}^k$ is globally convergent. 

Next, we prove the game admits a unique stationary distribution when no L-PNE exists. The following four cases are considered: (i) the game has CoACo edges; (ii) stubborn neighbors adopt the different strategies in some coordinator subnetwork;  (iii) no bipartition exists for an anti-coordinator subnetwork such that no edges connect players within the same subset; (iv) a bipartition $\{ACo_0^k, ACo_1^k\}$ (with no intra-subset edges) exists, but some anti-coordinators in one subset $ACo_0^k$ or $ACo_1^k$ have stubborn neighbors adopting both strategies 0 and 1.

For case (i), we only need to consider a connected subnetwork which possesses at least one CoACo edge after eliminating all stubborn agents. We first prove for any given strategy profile $(x_1,...,x_n) \in A^N$, there exists a positive probability that each coordinator or anti-coordinator $i$ in the subnetwork adopts any strategy $y_i \in A$, while other players keep their initial strategies. We analyze the following two cases separately: (i1) $i$ is located on a CoACo edge $ij$, and (i2) $i$ is not located on a CoACo edge. 

(i1) For $i \in Co$ and $y_i=x_j$, the probability that $i$ adopts $y_i$ and $k$ adopts $x_k$ for $\forall k \in N \backslash\{i\}$ is 
$$\big(p_i \prod_{e \neq i}\left(1-p_e\right)\big) \times\big(q_{i j} \prod_{f \neq j}\left(1-q_{i f}\right)\big)>0.$$ 

For $i \in Co$ and $y_i \neq x_j$, we only need consider the case $x_i \neq y_i$, for which $x_i = x_j$ (see Figure \ref{FIGS3} (a)). The probability that $j$ adopts $y_i$ and $k$ adopts $x_k$ for $\forall k \in N \backslash\{j\}$ (see Figure \ref{FIGS3} (b)) is
$$
\big(p_j \prod_{e \neq j}\left(1-p_e\right)\big) \times\big(q_{j i} \prod_{f \neq i}\left(1-q_{j f}\right)\big)>0.
$$
Based on the existing strategy profile (see Figure \ref{FIGS3} (b)), the probability that $i$ adopts $y_i$ (see Figure \ref{FIGS3} (c)) is
$$
\big(p_i \prod_{e \neq i}\left(1-p_e\right)\big) \times\big(q_{i j} \prod_{f \neq j}\left(1-q_{i f}\right)\big)>0.
$$
Based on the existing strategy profile (see Figure \ref{FIGS3} (c)), the probability that $j$ adopts $x_j$ (see Figure \ref{FIGS3} (d)) is
$$
\big(p_j \prod_{e \neq j}\left(1-p_e\right)\big) \times\big(q_{j i} \prod_{f \neq i}\left(1-q_{j f}\right)\big)>0.
$$
Thus, whether $x_j$ equals $y_i$ or not, the probability that $i$ changes $x_i$ to any $y_i \in A$ while keeping the strategies of other players unchanged (see Figure \ref{FIGS3} (a-d)) is positive.

\begin{figure}[H]
\centering
\includegraphics[width=0.75\textwidth]{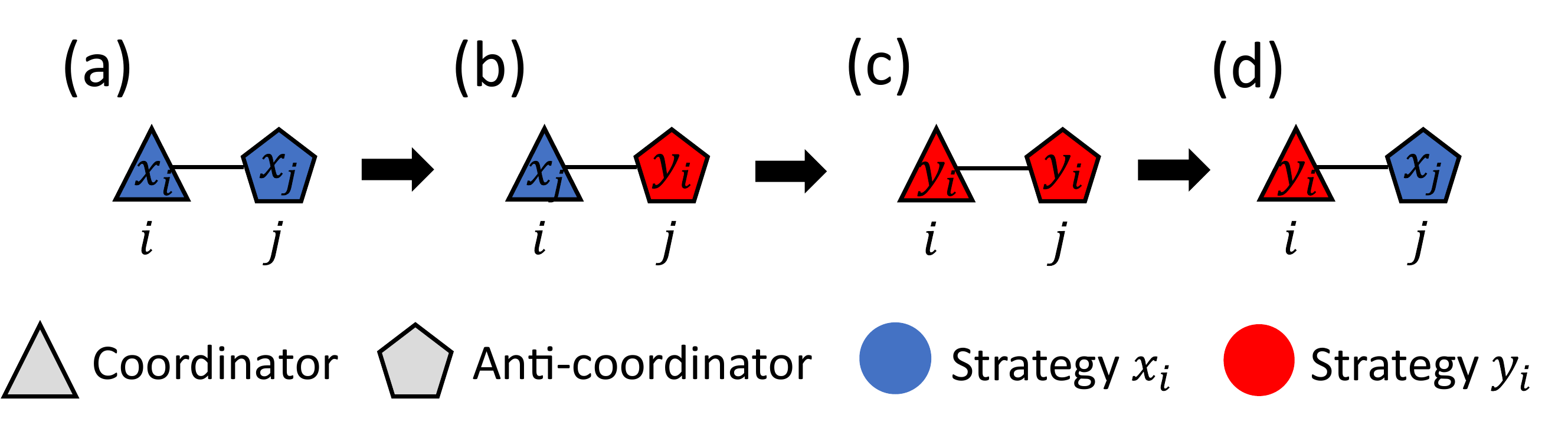}
\caption{A schematic diagram of the update path in the proof of case (i1) for $i \in Co$. Coordinators and anti-coordinators are represented by triangles and pentagons, respectively, and their strategies are $x_i$ and $x_j$. Each column ((a)-(d)) is an information path from $i$ to $j$.}\label{FIGS3}
\end{figure}

For $j \in ACo$ and $y_j \neq x_i$, the probability that $j$ adopts $y_j$ and $k$ adopts $x_k$ for $\forall k \in N \backslash\{j\}$ is
$$
\big(p_j \prod_{e \neq j}\left(1-p_e\right)\big) \times\big(q_{j i} \prod_{f \neq i}\left(1-q_{j f}\right)\big)>0.
$$

For $j \in ACo$ and $y_j = x_i$, we only need consider the case $x_j \neq y_j$, for which $x_i \neq x_j$ (see Figure \ref{FIGS4} (a)). The probability that $i$ adopts $x_j$ and $k$ adopts $x_k, \forall k \in N \backslash\{i\}$ (see Figure \ref{FIGS4} (b)) is
$$
\big(p_i \prod_{e \neq i}\left(1-p_e\right)\big) \times\big(q_{i j} \prod_{f \neq j}\left(1-q_{i f}\right)\big)>0.
$$
Based on the existing strategy profile (see Figure \ref{FIGS4} (b)), the probability that $j$ adopts $y_j$ (see Figure \ref{FIGS4} (c)) is
$$
\big(p_j \prod_{e \neq j}\left(1-p_e\right)\big) \times\big(q_{j i} \prod_{f \neq i}\left(1-q_{j f}\right)\big)>0.
$$
Based on the existing strategy profile (see Figure \ref{FIGS4} (c)), the probability that $i$ adopts $x_i$ (see Figure \ref{FIGS4} (d)) is
$$
\big(p_i \prod_{e \neq i}\left(1-p_e\right)\big) \times\big(q_{i j} \prod_{f \neq j}\left(1-q_{i f}\right)\big)>0.
$$
Thus, whether $x_i$ equals $y_j$ or not, the probability that $j$ changes $x_j$ to any $y_j \in$ $A$ while keeping the strategies of other players unchanged (see Figure \ref{FIGS4} (a-d)) is positive.
\begin{figure}[H]
\centering
\includegraphics[width=0.75\textwidth]{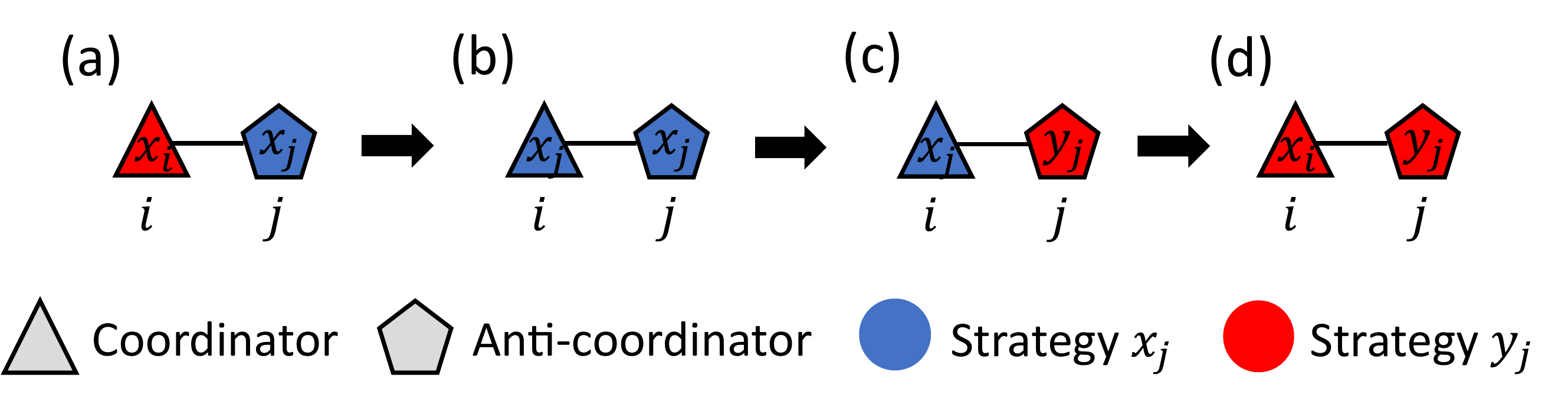}
\caption{A schematic diagram of the update path in the proof of case (i2) for $\forall j \in ACo$. Coordinators and anti-coordinators are represented by triangles and pentagons, respectively, and their strategies are $x_i$ and $x_j$. Each column ((a)-(d)) is an information path from $j$ to $i$.}\label{FIGS4}
\end{figure}

(i2) For $i \in Co$, there exists $j \in ACo$ and $\left\{c_1, \ldots, c_{l(i, j)}\right\} \subseteq Co$, such that $q_{i c_1} q_{c_1 c_2} \ldots q_{c_{l(i, j)} j}>0$, since the game has CoACo edges. This defines an information path from $i$ to $j$. In the later proof, we abbreviate $l(i, j)$ as $l$. We only consider the case $x_i=x_{c_1}=...=x_{c_l}=x_j$ (see Figure \ref{FIGS5} (a)), with analogous reasoning for other cases. The probability that $i$ adopts $x_i, c_k$ adopts $x_{c_k}, \forall k \in\{1, \ldots, l\}$, and $j$ adopts $y_i$ (see Figure \ref{FIGS5} (b)) is positive from the proof of (i1). Based on the existing strategy profile (see Figure \ref{FIGS5} (b)), the probability that $c_l$ adopts $y_i$ (see Figure \ref{FIGS5} (c)) is
$$
\big(p_{c_l} \prod_{e \neq c_l}\left(1-p_e\right)\big) \times\big(q_{c_l j} \prod_{f \neq j}\left(1-q_{c_l f}\right)\big)>0.
$$
Based on the existing strategy profile (see Figure \ref{FIGS5} (c)), the probability that $c_{l-1}$ adopts the strategy of $c_l$ is positive. Similarly, the probability that $c_{k-1}$ adopts the strategy of $c_k$ is also positive, $\forall k \in \{2, \ldots, l-1\}$. The same goes for $c_1$ and $i$. Therefore, the probability that all coordinators in this information path adopt $y_i$ (see Figure \ref{FIGS5} (d)) is positive. Based on the existing strategy profile (see Figure \ref{FIGS5} (d)), the probability that $j$ adopts $x_j$ (see Figure \ref{FIGS5} (e)) is positive from the proof of (i1). There clearly exists a positive probability that the strategy of adopting the right neighbor is adopted sequentially from $c_l$ to $c_1$ (see Figure \ref{FIGS5} (f)). Thus, the probability that $i \in Co$ changes $x_i$ to $y_i \in A$ while keeping the strategies of other players unchanged (see Figure \ref{FIGS5} (a-f)) is positive.

\begin{figure}[H]
\centering
\includegraphics[width=0.8\textwidth]{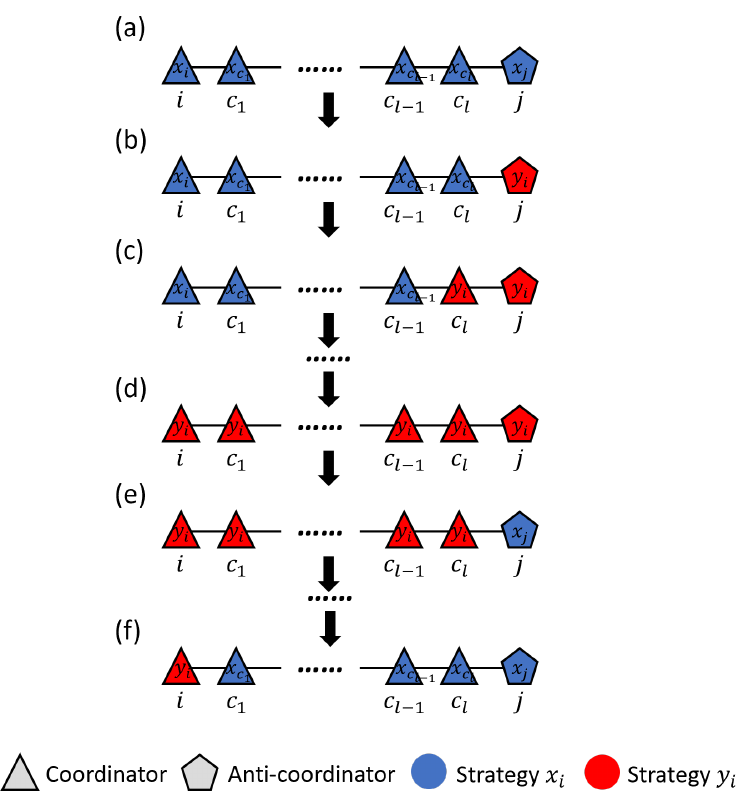}
\caption{A schematic diagram of the update path in the proof of case (i2) for $\forall i \in Co$. Coordinators and anti-coordinators are represented by triangles and pentagons, respectively, and their strategies are $ x_i, x_{c_1}, \ldots, x_{c_l}$, and $ x_j $. Each row ((a)-(f)) represents an information path from $i$ to $j$ passing through $\{c_1, \ldots, c_l\}$, except that the strategies of the players are different.}\label{FIGS5}
\end{figure}

For $j \in ACo$, there exists $i \in Co$ and $\left\{r_1, \ldots, r_{l}\right\} \subseteq ACo$, such that $q_{j r_1} q_{r_1 r_2} \ldots q_{r_l i}>0$, since the game has CoACo edges. This then defines an information path from $j$ to $i$. 
We only consider the case $l$ is odd, $x_j=x_{2k}$ for $\forall k \in \{ 1,...,(l-1)/2 \}$, and $x_{2k-1}=x_{2k+1}$ for $\forall k \in \{ 1,...,(l-1)/2 \}$ (see Figure \ref{FIGS6} (a)), with analogous reasoning for other cases. The probability that $j$ adopts $x_j$, $r_k$ adopts $x_{r_k}, \forall k \in\{1, \ldots, l\}$, and $i$ adopts $y_j$ (see Figure \ref{FIGS6} (b)) is positive from the proof of (i1). Based on the existing strategy profile (see Figure \ref{FIGS6} (b)), the probability that $r_l$ adopts $y_j' \in A \backslash \{ y_j \}$ (see Figure \ref{FIGS6} (c)) is
$$
\big(\,p_{r_l} \prod_{e \neq r_l} (1-p_e) \big) 
\times 
\big(q_{r_l i} \prod_{f \neq i} (1-q_{r_l f}) \big) > 0.
$$
Based on the existing strategy profile (see Figure \ref{FIGS6} (c)), the probability that $r_{l-1}$ adopts strategy $y_j$ is positive. Similarly, the probability that each $r_{k-1}$ adopts the different strategy of $r_k$ is positive, $\forall k \in \{2, \ldots, l-1\}$. The same goes for $r_1$ and $j$. Thus, the probability that $r_{2k-1}$ adopts $y_j'$ for $\forall k \in\{1, ...,(l+1) / 2\}$ while $i, j$ and $r_{2k}$ adopt $y_j$ for $\forall k \in\{1, ...,(l-1) /2\}$ (see Figure \ref{FIGS6} (d)) is also positive. Based on the existing strategy profile (see Figure \ref{FIGS6} (d)), we apply the aforementioned procedure sequentially to $r_1$ through $r_l$, thereby ensuring the initial strategies of other anti-coordinators remain unchanged (see Figure \ref{FIGS6} (e)). Similarly, the probability that $i$ adopts $x_i$ (see Figure \ref{FIGS6} (f)) is positive from the proof of (i1). Thus, the probability that $j \in ACo$ changes $x_j$ to $y_j \in A$ while keeping the strategies of other players unchanged (see Figure \ref{FIGS6} (a-f)) is positive.

\begin{figure}[H]
\centering
\includegraphics[width=0.8\textwidth]{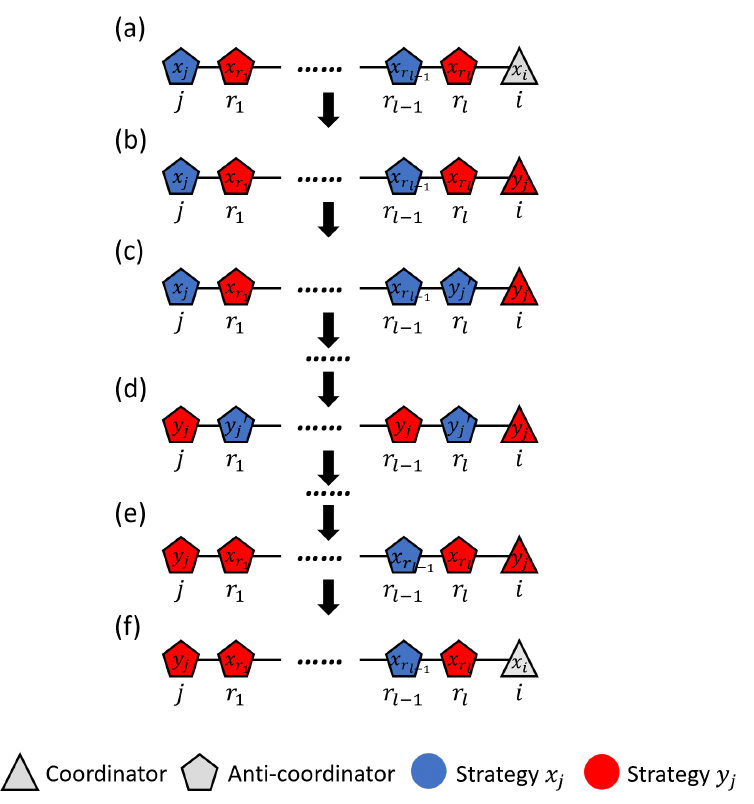}
\caption{A schematic diagram of the update path in the proof of case (i2) for $\forall j \in ACo$. Anti-coordinators and coordinators are represented by pentagons and triangles, respectively, and their strategies are $ x_j, x_{r_1}, \ldots, x_{r_l} $, and $x_i$. Each row ((a)-(f)) represents an information path from $ j $ to $ i $ passing through $ \{r_1, \ldots, r_l\} $, except that the strategies of the players are different.}\label{FIGS6}
\end{figure}

Combining cases (i1) and (i2), there exists a positive probability that all coordinators and anti-coordinators in the subnetwork containing the CoACo edges adopt any arbitrary strategy. Given the arbitrariness of $i \in Co \cup ACo$ and $y_i \in A$, we conclude that the strategy support is exactly $A$. This Markov chain is evidently finite-state and aperiodic, from which it follows that there exists a unique stationary distribution $\pi=\left(\pi_k\right)_{k \in Space}=$ $\left(m_{k k}^{-1}\right)_{k \in Space}$, where $m_{k k}=\sum_{n=1}^{\infty} n f_{k k}^{(n)}$ and $f_{k_1 k_2}^{(n)}$ denotes the probability of first returning to $k_2$ at step $n$ when starting from $k_1$. Thus, the game possesses a unique stationary distribution.

For case (ii), stubborn neighbors adopt the different strategies in some coordinator subnetwork. We invoke Theorem 2 in \cite{cao2024discrete} to establish that there always exists a positive probability for any coordinator to adopt an arbitrary strategy. This result further implies the existence of a unique stationary distribution.

For case (iii), no bipartition exists in the anti-coordinator subnetwork such that no edges connect any two players in the same subset, which implies the subnetwork is non-bipartite, and thus an odd-length cycle necessarily exists. Since an odd-length cycle cannot be 2-colored, for any strategy profile, there always exist a pair of anti-coordinator neighbors adopting the same strategy. As this pair of anti-coordinator neighbors updates their strategies simultaneously, and there is a positive probability that they only observe each other’s strategy information, there exists a positive probability for this pair to switch to any arbitrary strategy simultaneously. Noting that the anti-coordinator subnetwork is connected, we follow the update procedure illustrated in Figure \ref{FIGS6}, designating this pair of anti-coordinators as $i$ in the figure. By analogy, there exists a positive probability for all other anti-coordinators to adopt any arbitrary strategy, which further implies the existence of a unique stationary distribution.

For case (iv), the anti-coordinator subnetwork admits a bipartition $\{ACo_0^k, ACo_1^k\}$ with no intra-subset edges exists, yet some anti-coordinators in one subset (e.g., $ACo_0^k$) have stubborn neighbors adopting both strategies 0 and 1. In the scenario where there exists some $j \in ACo_0^k$ whose stubborn neighbors adopt distinct strategies, there is evidently a positive probability for $j$ to adopt any arbitrary strategy. Furthermore, there exists a positive probability for all other anti-coordinators whether directly or indirectly observing $j$ to adopt any arbitrary strategy as well. In the complementary scenario where no such $j \in ACo_0^k$ has stubborn neighbors with distinct strategies, the non-existence of an L-PNE implies that for any strategy profile, either there exists a pair of anti-coordinator neighbors adopting  the same strategy or there exists some anti-coordinator $r$ that adopts the same strategy as its some stubborn neighbor. The conclusion for the former case holds obviously, so we only need to prove the latter case. For any $i \in ACo$, if $i=r$, there exists a positive probability for $i$ to switch its strategy by solely observing the strategies of its stubborn neighbors. If $i \neq r$, there exists an information path of odd (or even) length through which $i \in ACo_0^k$ (or $ACo_1^k$) indirectly observes the strategies of stubborn neighbors adopting both Strategy 0 and Strategy 1. By following the update procedure illustrated in Figure \ref{FIGS6} and designating the stubborn neighbors as $j$ in the figure, there exists a positive probability for $i$ to adopt any arbitrary strategy. From this, it follows that a unique stationary distribution exists.
\end{proof}

Corollary \ref{co4} also holds for the synchronous best response dynamics. This is because under synchronous updating regime, there exists a positive probability that a player keeps his/her original strategy when failing to observe the strategies of any neighbors. This implies that update paths corresponding to asynchronous updating are constructible in the synchronous updating regime, and the conclusion follows naturally.

\section{Agent-based simulations}
\subsection{CRS game generation}
Consider a network with $n_C$ conformists, $n_R$ rebels, $n_{S_0}$ stubborn agents with dominant strategy $0$, and $n_{S_1}$ stubborn agents with dominant strategy $1$. In the following, we demonstrate how to add edges to this network such that the homophily and heterophily indices are sufficiently close to predefined $h_{XY}$ for all $X,Y \in \{C, R, S_0, S_1\}$, and control the network average degree $\langle d \rangle$. Following the idea of Lemma 1 in \cite{zhang2018fashion}, we use connection probabilities between different types of players to adjust the homogeneity and heterogeneity indices. Suppose that a type $X$ player and a type $Y$ player are connected with probability $p_{XY}$. Then the expected numbers of $XX$ edges and $XY$ edges are $K_{XX} = \binom{n_X}{2} p_{XX}$ and $K_{XY} = n_X n_Y p_{XY}$, respectively, where $n_X$ is the number of players with type $X$. 

By definition, we have
$$
h_{XX}= \frac{(n_X-1)p_{XX}}{(n_X-1)p_{XX} + \sum_{Z \in \{C, R, S_0, S_1\} \setminus \{X\}} n_Zp_{XZ}},
$$
$$
h_{XY}= \frac{n_Yp_{XY}}{(n_X-1)p_{XX} + \sum_{Z \in \{C, R, S_0, S_1\} \setminus \{X\}} n_Zp_{XZ}}.
$$
For given $n_X$, $h_{XY}$, and $\langle d \rangle$, we can take $p_{CC}=\frac{\gamma h_{CC}}{n_C-1}$, $p_{RR}=\frac{\gamma n_C h_{CR} h_{RR}}{n_R h_{RC} (n_R-1)}$, $p_{CR}=p_{RC}=\frac{\gamma h_{CR}}{n_R}$, $p_{CS_0}=\frac{\gamma h_{CS_0}}{n_{S_0}}$, $p_{CS_1}=\frac{\gamma h_{CS_1}}{n_{S_1}}$, $p_{RS_0}=\frac{\gamma n_C h_{CR} h_{RS_0}}{n_R h_{RC} n_{S_0}}$, and $p_{RS_1}=\frac{\gamma n_C h_{CR} h_{RS_1}}{n_R h_{RC} n_{S_1}}$ to make the expected homophily and heterophily indices of the generated network are exactly $h_{XY}$ for all $X,Y \in \{C, R, S_0, S_1\}$, where $\gamma = \frac{n h_{RC} \langle d \rangle }{n_C [ h_{CR} ( h_{RR} + 2h_{RC} + 2(h_{RS_0} + h_{RS_1}) ) + h_{RC} ( h_{CC} + 2h_{CS_0} + 2h_{CS_1} ) ]}>0$ is a scaling factor that modulates the average degree $\langle d \rangle$. 
Specifically, when $n$ is sufficiently large and  $n_C=n_R= n_{S_0}=n_{S_1}=\gamma$, we have $h_{XY} \approx p_{XY}$. 

In the simulations of Figure \ref{fig2} and Figure \ref{fig3}, we set $n_C = n_R = n_{S_0} = n_{S_1}=250$ for each subfigure and generate the networks multiple times to ensure that the homophily and heterophily indices of the generated network, denoted by $\hat h_{XY}$, satisfy $|\hat{h}_{XY} - h_{XY}| < 0.01$ for all $X,Y \in \{C, R, S_0, S_1\}$ (see S2.C and S2.D for details).


\subsection{HP game generation}
We generate an HP game based on its corresponding simplified CRS game. We first construct a CRS game via the CRS game generation method introduced above, then relabel conformists as coordinators and rebels as anti-coordinators, and finally set their payoff functions. Specifically, we assume that the payoff matrix for a player has the form
\begin{equation}\label{eq16}
\bordermatrix{
   & 0 & 1 \cr
0 & \varepsilon_0  & 0 \cr
1 & 0 & \varepsilon_1 \cr
},
\end{equation}
where $\varepsilon_0$ and $\varepsilon_1$ are independently and identically distributed random variables. For coordinators, we assume $\varepsilon_0$ and $\varepsilon_1$ independently follow a uniform distribution on $(z, 1)$ with $0<z<1$; for anti-coordinators, $\varepsilon_0$ and $\varepsilon_1$ independently follow a uniform distribution on $(-1, -z)$ with $0<z<1$. For stubborn agents with dominant strategy $0$ (or $1$), $\varepsilon_0$ is uniformly distributed on $(z, 1)$ (or $(-1, -z)$), and $\varepsilon_1$ is independently uniformly distributed on $(-1, -z)$ (or $(z, 1)$) with $0<z<1$. This completes the generation of the network HP game. 

For a given $z$, the expected value and variance of the threshold $\tau$ for coordinators and anti-coordinators can be easily derived. Proposition \ref{proposition1} shows that the threshold has an expected value of $\frac{1}{2}$ and its variance increasing monotonically as $z$ decreases.

\begin{proposition}\label{proposition1}
For a given $z \in (0,1)$, the threshold for coordinators and anti-coordinators has an expected value $\mathbb{E}(\tau(z)) = \frac{1}{2}$, and its variance is $D(\tau(z)) = \frac{3}{4} + \frac{1}{(1-z)^2} \left(z^2 \ln\frac{1+z}{2z} - \ln\frac{2}{1+z}\right)$ with $D'(\tau(z))<0$.
\end{proposition}

\begin{proof}
Recall that $\tau(z) = \frac{\varepsilon_1(z)}{\varepsilon_0(z) + \varepsilon_1(z)}$. By the symmetry of the i.i.d. random variables $\varepsilon_0$ and $\varepsilon_1$, it follows that 
$
\mathbb{E}(\tau) = \mathbb{E}\left(\frac{\varepsilon_1}{\varepsilon_0 + \varepsilon_1}\right) = \mathbb{E}\left(\frac{\varepsilon_0}{\varepsilon_0 + \varepsilon_1}\right).
$
Given $\frac{\varepsilon_1}{\varepsilon_0 + \varepsilon_1} + \frac{\varepsilon_0}{\varepsilon_0 + \varepsilon_1} = 1$, taking the expectation of both sides gives
$
\mathbb{E}\left(\frac{\varepsilon_1}{\varepsilon_0 + \varepsilon_1}\right) + \mathbb{E}\left(\frac{\varepsilon_0}{\varepsilon_0 + \varepsilon_1}\right) = 1
$,
and thus $\mathbb{E}(\tau) = \frac{1}{2}$.

From the variance formula $D(\tau) = \mathbb{E}(\tau^2) - \left[\mathbb{E}(\tau)\right]^2$, we have $D(\tau) = \frac{1}{(1-z)^2} \int_z^1 \int_z^1 \frac{\varepsilon_1^2}{(\varepsilon_0 + \varepsilon_1)^2} d\varepsilon_0 d\varepsilon_1 - \left(\frac{1}{2}\right)^2$. Substituting the evaluated double integral gives $D(\tau) = \frac{3}{4} + \frac{1}{(1-z)^2} \left(z^2 \ln\frac{1+z}{2z} - \ln\frac{2}{1+z}\right)$, and its derivative $D'(\tau) = \frac{2z\ln\frac{1+z}{2z}+\frac{(1-z)^2}{1+z}-2\ln\frac{2}{1+z}}{(1-z)^3}$. For all $z\in(0,1)$, the numerator of $D'(\tau)$ is negative and the denominator $(1-z)^3$ is positive, hence $D'(\tau)<0$.
\end{proof}

 
\subsection{Simulation details for Figure 2}
In the simulations of Figure 2, we set $n_C = n_R = n_{S_0} = n_{S_1}=250$, $h_{CS_0}=0.3$, $h_{CS_1}=h_{RS_0}=h_{RS_1}=0.1$, and the average degree $\langle d \rangle=16$ for all subfigures. The specific settings of the scaling parameter $\gamma$, edge probabilities $p_{XY}$ (for all $X,Y \in \{C, R, S_0, S_1\}$), and initial points for each subfigure are detailed below.

(A) Scaling parameter: $\gamma=32$. Edge probabilities: $p_{CC}=0.0643$, $p_{RR}=0.0386$, $p_{CR}=0.0128$, $p_{CS_0}=0.0384$, $p_{CS_1}=0.0128$, $p_{RS_0}=0.0064$, and $p_{RS_1}=0.0064$. Initial points: $(0.2,0.2)$, $(0.2,0.8)$, $(0.7,0.2)$, $(0.7,0.8)$.

(B)  Scaling parameter: $\gamma=34.0426$. Edge probabilities: $p_{CC}=0.0547$, $p_{RR}=0.0164$, $p_{CR}=0.0272$, $p_{CS_0}=0.0409$, $p_{CS_1}=0.0136$, $p_{RS_0}=0.0054$, and $p_{RS_1}=0.0054$. Initial points: $(0.2,0.2)$, $(0.2,0.3)$, $(0.2,0.4)$.

(C) Scaling parameter: $\gamma=24.6154$. Edge probabilities: $p_{CC}=0.0099$, $p_{RR}=0.0297$, $p_{CR}=0.0492$, $p_{CS_0}=0.0295$, $p_{CS_1}=0.0098$, $p_{RS_0}=0.0098$, and $p_{RS_1}=0.0098$. Initial points: $(0.2,0.2)$, $(0.2,0.3)$, $(0.2,0.4)$.

(D) Scaling parameter: $\gamma=30.6849$. Edge probabilities: $p_{CC}=0.0246$, $p_{RR}=0.0070$, $p_{CR}=0.0491$, $p_{CS_0}=0.0368$, $p_{CS_1}=0.0123$, $p_{RS_0}=0.0070$, and $p_{RS_1}=0.0070$. Initial points: $(0.1,0.1)$, $(0.2,0.2)$, $(0.3,0.3)$.

Based on the aforementioned parameters, we generate four CRS games with given homogeneity and heterogeneity indices. Subsequently, we employ the HP game generation method, setting $z = 0.5$ in the payoff matrix (\ref{eq16}) to further generate the HP games presented in the first row.

\subsection{Simulation details for Figure 3}
In the simulations of Figure 3, we set $n_C = n_R = n_{S_0} = n_{S_1}=250$, $h_{CS_0}=0.3$, and the average degree $\langle d \rangle=16$ for all subfigures. In subfigures B–D, $h_{CS_1}=h_{RS_0}=h_{RS_1}=0.1$, whereas in subfigure A these parameters are all set to 0. The specific settings of the scaling parameter $\gamma$ and edge probabilities $p_{XY}$ (for all $X,Y \in \{C, R, S_0, S_1\}$) are detailed below.

(A) Scaling parameter: $\gamma=0$. Edge probabilities: $p_{CC}=0.0782$, $p_{RR}=0.1117$, $p_{CR}=0$, $p_{CS_0}=0.0335$, $p_{CS_1}=0$, $p_{RS_0}=0$, and $p_{RS_1}=0$.

(B) Scaling parameter: $\gamma=32$. Edge probabilities: $p_{CC}=0.0643$, $p_{RR}=0.0386$, $p_{CR}=0.0128$, $p_{CS_0}=0.0384$, $p_{CS_1}=0.0128$, $p_{RS_0}=0.0064$, and $p_{RS_1}=0.0064$.

(C) Scaling parameter: $\gamma=24.6154$. Edge probabilities: $p_{CC}=0.0099$, $p_{RR}=0.0297$, $p_{CR}=0.0492$, $p_{CS_0}=0.0295$, $p_{CS_1}=0.0098$, $p_{RS_0}=0.0098$, and $p_{RS_1}=0.0098$.

(D) Scaling parameter: $\gamma=30.6849$. Edge probabilities: $p_{CC}=0.0246$, $p_{RR}=0.0070$, $p_{CR}=0.0491$, $p_{CS_0}=0.0368$, $p_{CS_1}=0.0123$, $p_{RS_0}=0.0070$, and $p_{RS_1}=0.0070$.

\subsection{Error analysis}

\subsubsection*{HP game $\&$ CRS game}
To test the robustness of the simplified method, we compare the equilibrium strategy frequencies of the HP game and its simplified CRS game across diverse networks. We adopt network structures with the same number of players and homophily indices as those in Figure 2, i.e., the structures specified by the quadruples $(h_{CC}, h_{RR}, h_{CR}, h_{RC}) = (0.5, 0.6, 0.1, 0.2)$, $(0.4, 0.3, 0.2, 0.5)$, $(0.2, 0.1, 0.4, 0.7)$, and $(0.1, 0.3, 0.5, 0.5)$, which yield 2, 1, 0, and 1 stable fixed points, respectively. For each network, we simulate both the HP game and its simplified CRS game from 16 different initial conditions $(x_c,x_r)$, where $x_c$ and $x_r$ both evenly sampled from 0.2 to 0.8 (i.e., each takes 4 values). For a given initial point $(x_c,x_r)$, the difference between the fixed points is defined as the absolute error $\varepsilon(x_{c}, x_{r}) = \max\{|\tilde{x}_c - \bar{x}_c|, |\tilde{x}_r - \bar{x}_r|\}$, where $(\tilde{x}_c, \tilde{x}_r)$ and $(\bar{x}_c, \bar{x}_r)$ are the convergence points of HP game and CRS game, respectively. Finally, the error for each network is taken as the average of all initial points.

We adjust the range of uniform distribution followed by the diagonal values in the payoff matrix (\ref{eq16}) to modulate the variance of the threshold, and analyzing the impact of threshold variance on approximation errors in the simplified CRS game. Specifically, we set $z$ to 0, 0.1, 0.2, 0.3, 0.4, and 0.5, with the corresponding variances (calculated based on Proposition \ref{proposition1}) of 0.0569, 0.0330, 0.0205, 0.0129, 0.0080, and 0.0047 respectively. Figure \ref{FIGS7} shows the variation of error values in approximating HP games with CRS games, as the threshold variance changes across different networks. As shown in the figure, the error value tends to increase with increasing variance. Furthermore, for large values of variance, the error values for the network with two stable fixed points (red dots) are significantly higher than those for other network.

\begin{figure}[H]
\centering
\includegraphics[width=1\textwidth]{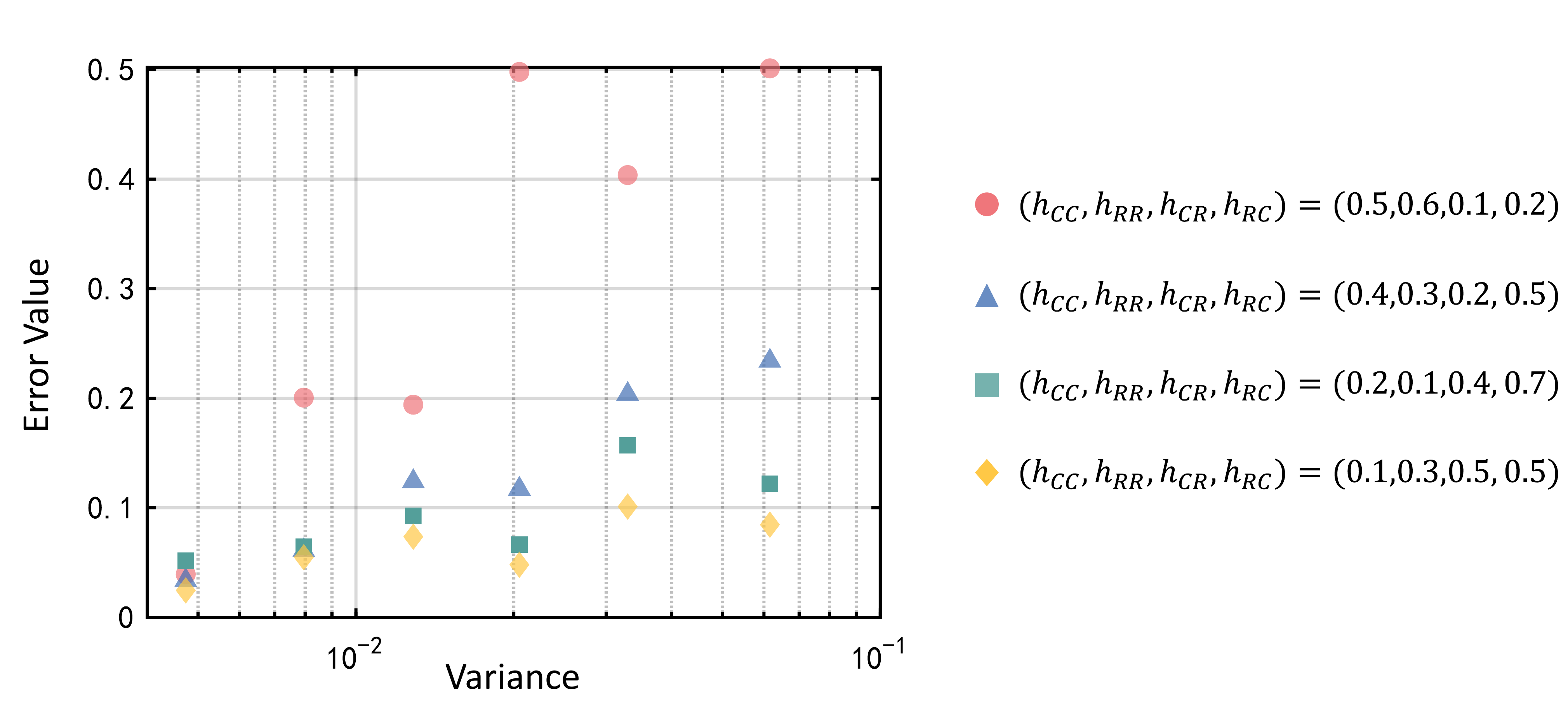}
\caption{Differences between the equilibrium strategy frequencies of the HP game and its simplified CRS game across diverse networks. The $x$-axis (i.e., the variance of the threshold) is presented on a logarithmic scale. Red dots, blue triangles, green squares, and yellow diamonds denote the parameter sets $(h_{CC}, h_{RR}, h_{CR}, h_{RC}) = (0.5, 0.6, 0.1, 0.2)$, $(0.4, 0.3, 0.2, 0.5)$, $(0.2, 0.1, 0.4, 0.7)$, and $(0.1, 0.3, 0.5, 0.5)$, respectively.}\label{FIGS7}
\end{figure}

\subsubsection*{CRS game $\&$ ODE (\ref{eq2})}
To test the robustness of the deterministic approximation method, we calculate the difference between the fixed points of the CRS game and ODE (\ref{eq2}) for different combinations of homophily and heterophily indices \citep[see e.g.][]{PEI2024dynamic}. Specifically, we fix two sets of heterophily indices: the first is $(h_{CS_0},h_{CS_1},h_{RS_0},h_{RS_1})=(0.2,0.1,0.1,0.1)$, and the second is $(h_{CS_0},h_{CS_1},h_{RS_0},h_{RS_1})=(0.3,0.1,0.1,0.1)$. For the first (second) set, $h_{CC}$ is varied from 0 to 0.6 (0.5) with interval 0.1, whereas $h_{RR}$ is varied from 0 to 0.7 with interval 0.1 (so their are 56 or 48 combinations in total). For each combination, we generate a CRS network and run from 16 different initial points $(x_c,x_r)$, where $x_c$ and $x_r$ are evenly sampled from the set $\{0.2, 0.4, 0.6, 0.8\}$. For each initial point $(x_c,x_r)$, the difference between the fixed points of ODE (\ref{eq2}) and CRS game is defined as the absolute error $\varepsilon(x_{c}, x_{r}) = \max\{|\hat{x}_c - \bar{x}_c|, |\hat{x}_r - \bar{x}_r|\}$, where $(\hat{x}_c, \hat{x}_r)$ and $(\bar{x}_c, \bar{x}_r)$ are the convergence points of ODE (\ref{eq2}) and CRS game, respectively. Finally, the error for a combination $(h_{CC}, h_{RR})$ (i.e., a grid point in Figure \ref{FIGS9}) is taken as the average error of its 16 initial points.

We first specify stable fixed points of ODE (\ref{eq2}) for different combinations of homophily and heterophily indices. In Figure \ref{FIGS8} (A), the stable fixed points in each region are as follows. Region I: $(0,1)$, $(1,0)$; Region II: $(0,\frac{x_r^*}{h_{RR}})$, $(1,\frac{x_r^*-h_{RC}}{h_{RR}})$; Region III: $(1,0)$; Region IV: $(1,0)$; Region V: $(1,0)$, $(\bar{x}_c,\bar{x}_r)$; Region VI: $(1,\frac{x_r^*-h_{RC}}{h_{RR}})$; Region VII: None; Region VIII: $(\bar{x}_c,\bar{x}_r)$. In Figure \ref{FIGS8} (B), the stable fixed points in each region are as follows. Region I: $(0,1)$, $(1,0)$; Region II: $(0,\frac{x_r^*}{h_{RR}})$, $(1,\frac{x_r^*-h_{RC}}{h_{RR}})$; Region III: $(1,0)$; Region IV: $(1,0)$; Region V: $(1,0)$, $(\bar{x}_c,\bar{x}_r)$; Region VI: $(1,\frac{x_r^*-h_{RC}}{h_{RR}})$, $(\bar{x}_c,\bar{x}_r)$; Region VII: $(1,\frac{x_r^*-h_{RC}}{h_{RR}})$; Region VIII: None; Region IX: $(\bar{x}_c,\bar{x}_r)$; Region X: None.

\begin{figure}[H]
\centering
\includegraphics[width=1\textwidth]{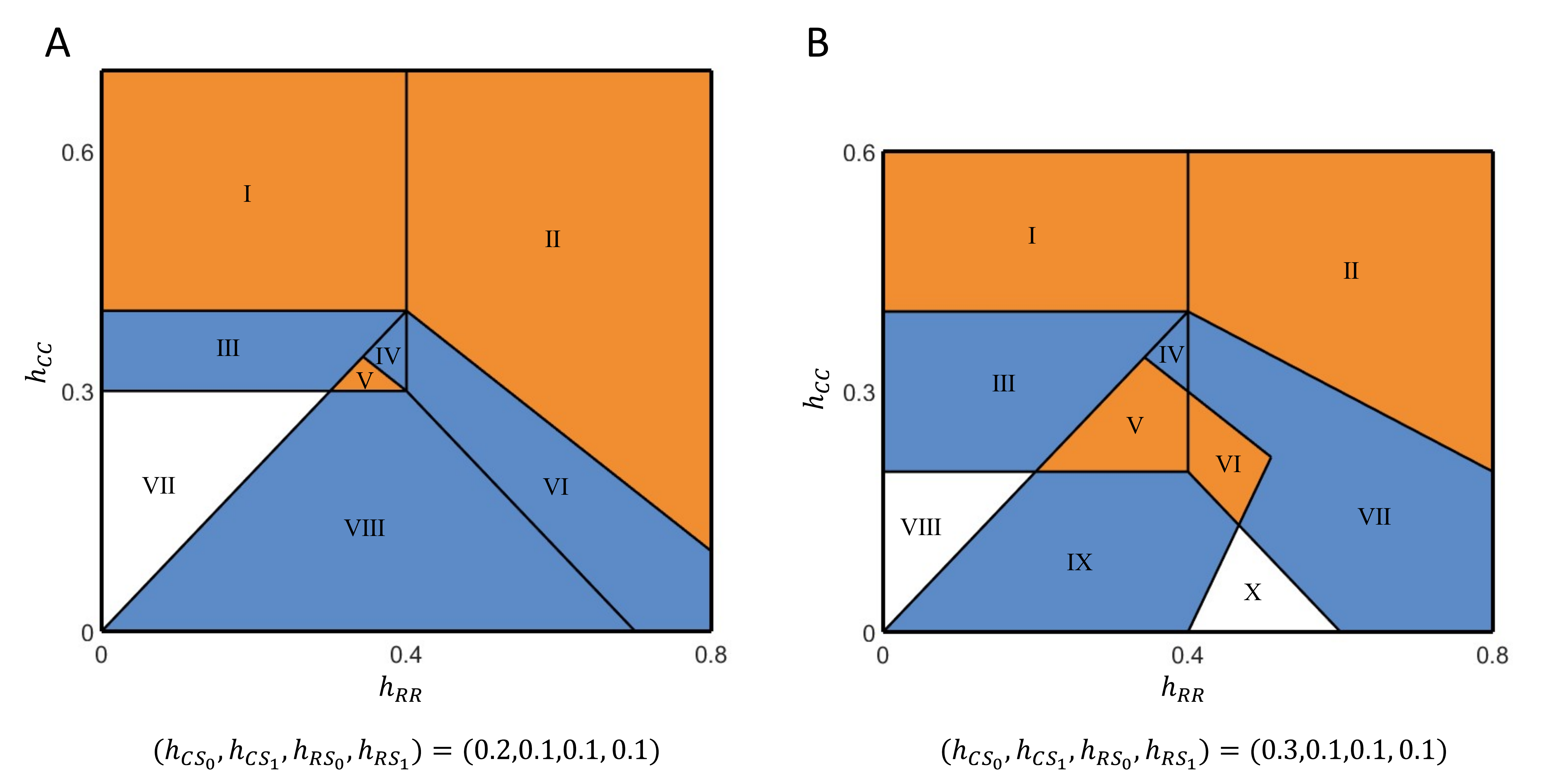}
\caption{Number of stable fixed points of ODE (\ref{eq2}) for different network structures. Parameters are taken as (A) $(h_{CS_0},h_{CS_1},h_{RS_0},h_{RS_1})=(0.2,0.1,0.1,0.1)$, (B) $(h_{CS_0},h_{CS_1},h_{RS_0},h_{RS_1})=(0.3,0.1,0.1,0.1)$. There is no stable fixed point in white region, one stable fixed point in blue regions, and two stable fixed points in orange regions.}\label{FIGS8}
\end{figure}

The first row of Figure \ref{FIGS9} shows the (total) error for the fixed point comparison. Three observations can be obtained. First, increasing the population size can slightly decrease the error. This is consistent with the stochastic approximation analysis, where the error term decreases in the population size. Second, the error is increasing in the number of stable fixed points. Specifically, two stable fixed points coexist for large values of $h_{CC}$ (top right regions in the first row of Figure \ref{FIGS9}), and the errors in these regions are also larger than in other regions. This may be because the trajectories of ODE (\ref{eq2}) are more likely to converge to `wrong' fixed points when there are many stable fixed points. This conjecture has been verified by separating the first and second types of errors (see the second and third rows of Figure \ref{FIGS9}). Here we use a simple rule to separate these two types of error: if the absolute error $\varepsilon$ is less (or greater) than a threshold, then we categorize it into the first (or second) type \citep[see e.g.][]{PEI2024dynamic}. Since the absolute distance between two stable fixed points is 1, the threshold is taken as 0.5. We find that the first type of error is smaller for large population sizes. In contrast, the second type of errors is increasing in the number of stable fixed points. Finally, errors are relatively large for initial points that are close to the lines $h_{CC} + h_{RR} = 0.8$ (or $2h_{CC} + h_{RR} = 1.2$). This is because the homophily indices on these lines are the bifurcation values of ODE (\ref{eq2}). Therefore, the trajectories of the CRS game are very sensitive to stochastic factors.

\begin{figure}[H]
\centering
\includegraphics[width=1\textwidth]{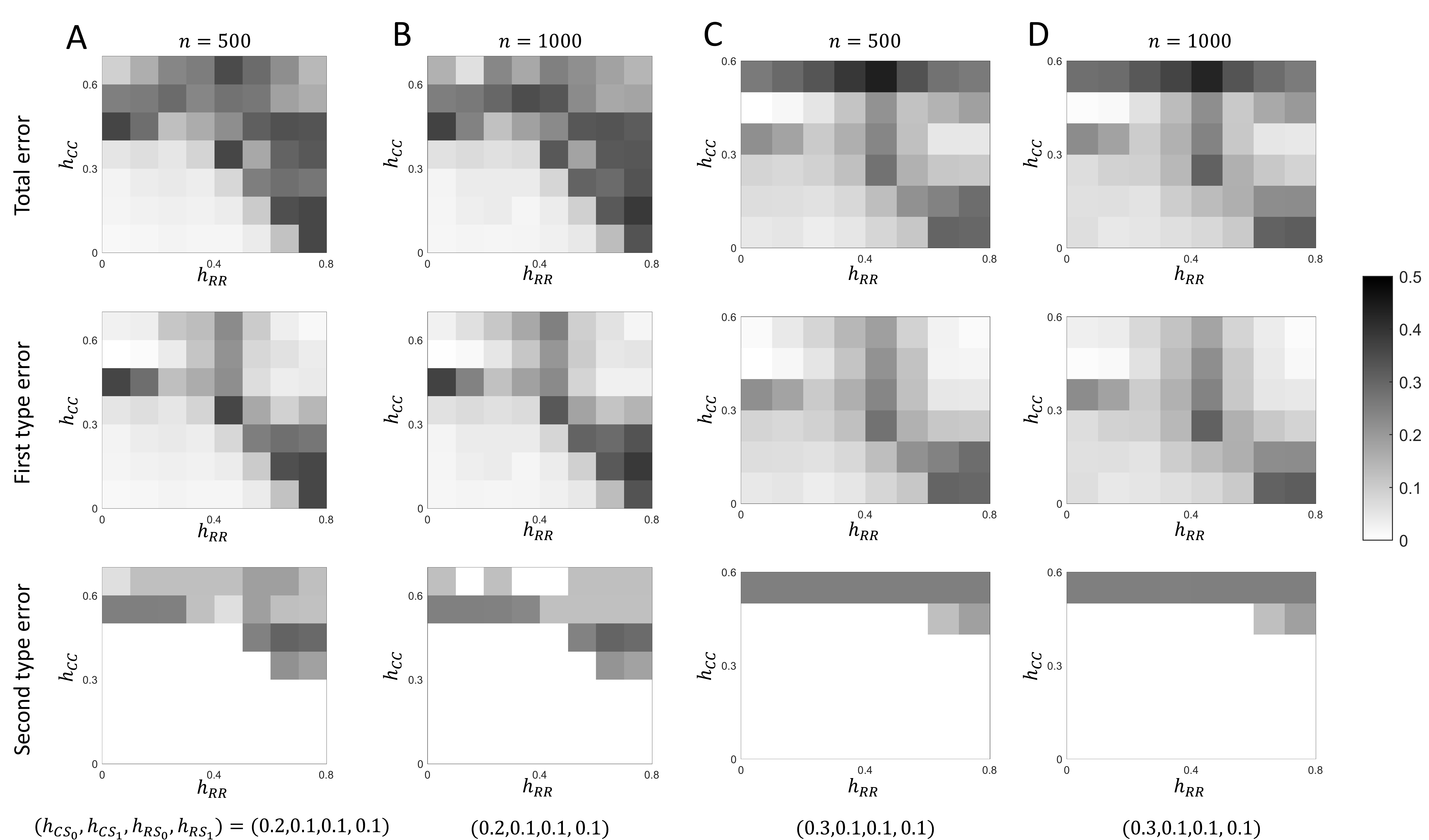}
\caption{Comparisons between fixed points of CRS game and ODE (\ref{eq2}). The population size in (A) and (C) are taken as $n = 500$ and (B) and (D) are taken as $n = 1000$. (A) and (B) $(h_{CS_0},h_{CS_1},h_{RS_0},h_{RS_1})=(0.2,0.1,0.1,0.1)$ and (C) and (D) $(h_{CS_0},h_{CS_1},h_{RS_0},h_{RS_1})=(0.3,0.1,0.1,0.1)$. The average degree is 16 for each network. From top to bottom, total error, first type of error, and second type of error in fixed point comparison. Errors at a grid point $(h_{CC}, h_{RR})$ are indicated by shading, whose scale is given to the right of the columns. From left to right, the average errors of the first row are $0.1758$, $0.1748$, $0.1606$, and $0.1601$, respectively; the average errors of the second row are $0.1101$, $0.1153$, $0.1130$, and $0.1120$, respectively; the average errors of the third row are $0.0657$, $0.0595$, $0.0482$, and $0.0481$, respectively.}\label{FIGS9}
\end{figure}

\section{Application: Prisoner's Dilemma game on networks with preference heterogeneous players}

The real dataset used in our study was collected by Ruiz-García et al. across 13 schools in diverse regions of Spain \citep[see e.g.][]{Miguel2023}, encompassing 3395 students and 60,566 declared relationships. From the 13 schools considered, 3 of them are in the Region of Madrid and the rest are in Andalucía. In this section, we adopt their notation, i.e. schools in Madrid are denoted as t1, t2 and t6, while those in Andalucía are denoted as t11\_1, ..., t11\_10. In their data set, each student could rate their relationship with any other student in their school as very good $(+2)$, good $(+1)$, bad $(-1)$ or very bad $(-2)$, and these positive or negative relationship ratings were mapped to the weight of the relationship $w_{ij}$ in the main text. The prosociality index of each student was mapped to the altruism parameter $\alpha_i$ in the main text, where prosociality was evaluated for each student through the answer to the three questions about sharing: 

\begin{itemize}
\item What do you prefer? A) €10 for yourself and €10 for your partner ($\mathcal{Q}_1 = 1$); B) €10 for yourself and €0 for your partner ($\mathcal{Q}_1 = 0$).
\item What do you prefer? A) €10 for yourself and €10 for your partner ($\mathcal{Q}_2 = 0$); B) €10 for yourself and €20 for your partner ($\mathcal{Q}_2 = 1$).
\item What do you prefer? A) €10 for yourself and €10 for your partner ($\mathcal{Q}_3 = 1$); B) €20 for yourself and €0 for your partner ($\mathcal{Q}_3 = 0$).
\end{itemize}
The prosociality index was then computed as $(\mathcal{Q}_1 + \mathcal{Q}_2 + \mathcal{Q}_3)/3$, yielding a scalar value of 0, 0.33, 0.67, or 1.

In subsequent analyses, we set $b = 2$ and $\beta = 1/4$, and classified all students in the dataset using the method presented in the main text. Under this parameter set, students with $\alpha_i = 1$ are stubborn cooperators, while those with $\alpha_i = 0$ are stubborn defector, and students with all other values of $\alpha_i$ are either conformists or rebels. Table \ref{TabS1} summarizes the numbers of player types and edge types in the simplified CRS game across the 13 schools. The results indicate that the vast majority of students are stubborn agents, with rebels only in three schools: t11\_5, t11\_8, and t11\_9 (corresponding to School 1, 2, and 3 in the main text, respectively).

\begin{table}
\centering
\setlength{\tabcolsep}{2pt}
\caption{\textbf{Number of Player Types and Edge Types Across 13 Social Networks} (Strategy 0 = Cooperation, Strategy 1 = Defection).}\label{TabS1}
\begin{tabular}{c c c c c c c c c c c c c c}
\toprule
 & t11\_1 & t11\_2 & t11\_3 & t11\_4 & t11\_5 & t11\_6 & t11\_7 & t11\_8 & t11\_9 & t11\_10 & t1 & t2 & t6 \\
\midrule
\textbf{$n_C$} & 5 & 27 & 6 & 3 & 7 & 5 & 6 & 43 & 47 & 12 & 1 & 3 & 9 \\
\textbf{$n_R$} & 0 & 0 & 0 & 0 & 1 & 0 & 0 & 2 & 1 & 0 & 0 & 0 & 0 \\
\textbf{$n_{S_0}$} & 180 & 379 & 115 & 77 & 167 & 57 & 63 & 128 & 209 & 276 & 313 & 176 & 477 \\
\textbf{$n_{S_1}$} & 47 & 106 & 35 & 30 & 48 & 44 & 11 & 35 & 62 & 98 & 95 & 59 & 48 \\
Total Nodes & 232 & 512 & 156 & 110 & 223 & 106 & 80 & 208 & 319 & 386 & 409 & 238 & 534 \\
\midrule
\textbf{$K_{CC}$} & 0 & 25 & 4 & 1 & 7 & 3 & 5 & 72 & 66 & 10 & 0 & 0 & 5 \\
\textbf{$K_{CR}$} & 0 & 0 & 0 & 0 & 1 & 0 & 0 & 10 & 5 & 0 & 0 & 0 & 0 \\
\textbf{$K_{CS_0}$} & 97 & 488 & 120 & 39 & 136 & 56 & 100 & 393 & 546 & 234 & 115 & 84 & 354 \\
\textbf{$K_{CS_1}$} & 20 & 147 & 42 & 13 & 53 & 160 & 27 & 101 & 170 & 77 & 27 & 13 & 39 \\
\textbf{$K_{RR}$} & 0 & 0 & 0 & 0 & 0 & 0 & 0 & 0 & 0 & 0 & 0 & 0 & 0 \\
\textbf{$K_{RS_0}$} & 0 & 0 & 0 & 0 & 32 & 0 & 0 & 8 & 32 & 0 & 0 & 0 & 0 \\
\textbf{$K_{RS_1}$} & 0 & 0 & 0 & 0 & 16 & 0 & 0 & 2 & 10 & 0 & 0 & 0 & 0 \\
Total Edges & 3293 & 7095 & 2388 & 990 & 2725 & 1293 & 813 & 1416 & 3027 & 4955 & 6509 & 2889 & 9527 \\
\bottomrule
\end{tabular}
\end{table}

We further analyze the three networks in which conformists and rebels coexist. By calculating the homophily indices (see Table \ref{TabS2}) and combining the results with Table 1 in the main text, we find that all three networks possess a unique stable fixed point $(1,0)$ of ODE (\ref{eq2}). Thus, our deterministic approximation method provides a simple way to predict the level of cooperation in the network Prisoner’s Dilemma game with preference heterogeneous players: the proportion of cooperation equals the sum of the proportions of conformists and stubborn cooperators. We then perform agent-based simulations to compute the proportion of cooperation in the HP game and its simplified CRS game, where the initial strategies of stubborn agents correspond precisely to their dominant strategies, while those of conformists and rebels are randomly assigned, with $10^5$ time steps run for convergence. The simulation results (see Table \ref{TabS2}) show the predictions of ODE (\ref{eq2}) differ by less than 0.01 from those of the simplified CRS game and by less than 0.05 from those of the HP game.

\begin{table}[H]
\centering
\caption{\textbf{Homophily Indices, Fixed Points, and Cooperation Proportion} (Strategy 0 = Cooperation, Strategy 1 = Defection).}
\label{TabS2}
\begin{tabular}{c c c c}
\toprule
  & t11\_5 & t11\_8 & t11\_9 \\
\midrule
\multicolumn{4}{l}{\textbf{Homophily Indices}} \\
$h_{CC}$ & 0.0686 & 0.2222 & 0.1547 \\
$h_{CR}$ & 0.0049 & 0.0154 & 0.0059 \\
$h_{CS_0}$ & 0.6667 & 0.6065 & 0.6401 \\
$h_{CS_1}$ & 0.2598 & 0.1559 & 0.1993 \\
$h_{RR}$ & 0 & 0 & 0 \\
$h_{RC}$ & 0.0204 & 0.5000 & 0.1064 \\
$h_{RS_0}$ & 0.6531 & 0.4000 & 0.6809 \\
$h_{RS_1}$ & 0.3265 & 0.1000 & 0.2128 \\
\addlinespace
\multicolumn{4}{l}{\textbf{Fixed Points}} \\
Stable Fixed Points & (1,0) & (1,0) & (1,0) \\
Unstable Fixed Points & None & None & None \\
\addlinespace
\multicolumn{4}{l}{\textbf{Cooperation Proportion}} \\
HP Game & 0.7760 & 0.7700 & 0.7540 \\
Simplified CRS Game & 0.7780 & 0.8150 & 0.7980 \\
ODE (2) & 0.7800 & 0.8220 & 0.8030 \\
\bottomrule
\end{tabular}
\end{table}


\clearpage 



\end{document}